\documentclass{article} 
\usepackage{iclr2025_conference,times}

\usepackage{amsmath,amsfonts,bm}

\def\eqref#1{equation~\ref{#1}}

\def\1{\bm{1}}

\DeclareMathAlphabet{\mathsfit}{\encodingdefault}{\sfdefault}{m}{sl}
\SetMathAlphabet{\mathsfit}{bold}{\encodingdefault}{\sfdefault}{bx}{n}

\newcommand{\softmax}{\mathrm{softmax}}

\usepackage{hyperref}
\usepackage{url}
\usepackage{xspace}
\usepackage{amsthm}
\usepackage{amsmath}
\usepackage{yfonts}
\usepackage{todonotes}

\title{Transformer Encoder Satisfiability: Complexity and Impact on Formal Reasoning}

\author{Marco S\"alzer \\
			University of Kaiserslautern-Landau, RPTU \\
			Kaiserslautern, Germany\\
	\texttt{marco.saelzer@rptu.de} 
	\And
	Eric Alsmann \\
	Theoretical Computer Science / Formal Methods\\
	University of Kassel, \ 
	Germany\\
	\texttt{eric.alsmann@uni-kassel.de} 
	\And
	Martin Lange \\
	Theoretical Computer Science / Formal Methods\\
	University of Kassel, \ 
	Germany\\
	\texttt{martin.lange@uni-kassel.de} 
}

\iclrfinalcopy 
\begin{document}

\newcommand*{\nats}{\mathbb{N}}
\newcommand*{\rats}{\mathbb{Q}}
\newcommand*{\reals}{\mathbb{R}}
\newcommand*{\bs}[1]{\boldsymbol{#1}}
\newcommand*{\stack}{\mathit{stack}}
\newcommand*{\hmax}{\mathrm{hardmax}}
\newcommand*{\relu}{\mathit{relu}}
\newcommand*{\scalar}[1]{\langle #1\rangle}
\newcommand{\countt}[2]{|#2|_{#1}}

\newcommand{\len}{\mathit{len}}

\newcommand*{\fixedfracrats}{\mathbb{Q}^\infty_f}
\newcommand*{\fpaset}{\mathbb{Q}^b_f}
\newcommand*{\satu}{\mathit{satu}}
\newcommand*{\wrap}{\mathit{wrap}}
\newcommand*{\rup}{{\uparrow}}

\newcommand*{\downsatu}{F^\downarrow_\satu}
\newcommand*{\downwrap}{F^\downarrow_\wrap}
\newcommand*{\upwrap}{F^\uparrow_\wrap}
\newcommand{\upsatu}{F^\uparrow_\satu}

\newcommand*{\rdown}[1]{\lfloor #1 \rfloor}
\newcommand*{\pwlfnn}{\mathcal{N}(\relu)}
\newcommand*{\gad}[1]{\langle #1\rangle}
\newcommand*{\para}{|\!|}

\newcommand*{\emb}{\mathit{emb}}
\newcommand*{\pos}{\mathit{pos}}
\newcommand*{\att}{\mathit{att}}
\newcommand*{\score}{\mathit{score}}
\newcommand*{\norm}{\mathit{norm}}
\newcommand*{\pool}{\mathit{pool}}
\newcommand*{\comb}{\mathit{comb}}
\newcommand*{\out}{\mathit{out}}

\newcommand*{\transC}{\mathcal{T}}
\newcommand*{\transCundec}{\mathcal{T}_\mathit{udec}}
\newcommand*{\logtransCundec}{\mathcal{T}^\textsc{log}_\mathit{udec}}
\newcommand*{\transCfix}{\mathcal{T}^\textsc{fix}}
\newcommand*{\periodTransCfix}{\mathcal{T}^\textsc{fix}_\circ}

\newcommand*{\empt}{\textsc{trSat}}
\newcommand*{\bempt}{\textsc{btrSat}}
\newcommand*{\pcp}{\textsc{Pcp}}
\newcommand*{\unbOTP}{\ensuremath{\textsc{OTP}^*}\xspace}
\newcommand*{\bOTP}{\ensuremath{\textsc{OTP}}\xspace}
\newcommand*{\unbOTWP}{\ensuremath{\textsc{OTWP}^*}\xspace}
\newcommand*{\bOTWP}{\ensuremath{\textsc{OTWP}}\xspace}
\newcommand*{\expbOTWP}{\ensuremath{\textsc{OTWP}^{\mathsf{exp}}}\xspace}
\newcommand*{\expbSQTP}{\ensuremath{\textsc{SQTP}^{\mathsf{exp}}}\xspace}
\newcommand*{\polybOTWP}{\ensuremath{\textsc{OTWP}^{\mathsf{poly}}}\xspace}

\newcommand{\size}{\mathit{size}}
\newcommand{\poly}{\mathit{poly}}
\newcommand{\NEXPTIME}{\text{NEXPTIME}\xspace}
\newcommand*{\ntime}[1]{\textsc{N$#1$-TIME}}

\newtheorem{theorem}{Theorem}
\newtheorem{lemma}{Lemma}
\newtheorem{corollary}{Corollary}


\newcommand{\TODOcomment}[2]{%
  \stepcounter{TODOcounter#1}%
  {\scriptsize\bf$^{(\arabic{TODOcounter#1})}$}%
  \marginpar[\fbox{
    \parbox{2cm}{\raggedleft
      \scriptsize$^{({\bf{\arabic{TODOcounter#1}{#1}}})}$%
      \scriptsize #2}}]%
  {\fbox{\parbox{2cm}{\raggedright
      \scriptsize$^{({\bf{\arabic{TODOcounter#1}{#1}}})}$%
      \scriptsize #2}}}
}%

\newcommand{\simpleTODOcomment}[2]{%
  \stepcounter{TODOcounter#1}%
  {\bf
    \scriptsize({\arabic{TODOcounter#1}~{#1}})
    {\bfseries{TODO:} #2}
  }
}

\newcounter{TODOcounter}
\newcommand{\TODO}[1]{\TODOcomment{}{#1}}
\newcommand{\TODOX}[1]{\simpleTODOcomment{}{#1}}

\maketitle

\begin{abstract}
We analyse the complexity of the satisfiability problem, or similarly feasibility problem, 
(trSAT) for transformer encoders (TE), which naturally occurs in formal verification or interpretation, collectively referred to as formal reasoning. 
We find that trSAT is undecidable when considering TE as they are commonly studied in the expressiveness community. Furthermore, we identify practical 
scenarios where trSAT is decidable and establish corresponding complexity bounds. Beyond trivial cases, we find that quantized TE, those restricted by 
fixed-width arithmetic, lead to the decidability of trSAT due to their limited attention capabilities. However, the problem remains difficult, as we 
establish scenarios where trSAT is NEXPTIME-hard and others where it is solvable in NEXPTIME for quantized TE. To complement our complexity results, 
we place our findings and their implications in the broader context of formal reasoning.
\end{abstract}

\section{Introduction}
\label{sec:intro}

Natural language processing (NLP) models, processing and computing human language, 
are gateways for modern applications aiming to interact with human users in a natural way. 
Although NLP is a traditional field of research, the use of deep learning techniques has 
undoubtedly revolutionised the field in recent years (\cite{OtterMK21}).  
In this revolution, models such as Recurrent Neural Networks (RNN) or more specific Long 
Short-term Memory Networks (LSTM) (\cite{YuSHZ19})
 have long been the driving force, but for a few years 
now NLP has a new figurehead: \emph{transformers} (\cite{VaswaniSPUJGKP17}).

Transformers are a deep learning model using (multiple) self-attention mechanisms to
process sequential input data, usually natural language. The efficient trainability 
of transformers, for example in contrast to LSTM, while achieving top-tier performance 
led to numerous heavy-impact implementations such as BERT (\cite{DevlinCLT19}), GPT-3 
(\cite{BrownMRSKDNSSAA20}) or GPT-4 (\cite{openai2023gpt4}), sparking
widespread use of the transformer architecture. However, the foreseeable omnipresence 
of transformer-based applications leads to serious security concerns.

In general, there are two approaches to establishing trustworthiness of learning-based models: 
first, certifying specific, application-dependent safety properties, called
\emph{verification}, and second, interpreting the behaviour of such models 
and giving explanations for it, called \emph{interpretation}. In both approaches, 
the holy grail is to develop automatic methods that are \emph{sound and complete}: 
algorithm $A$ given some model $T$ and (verification or interpretation) specification $\varphi$ outputs 
\texttt{true} if $T$ satisfies $\varphi$ (soundness) and for every given pair $T, \varphi$ where $T$ satisfies
$\varphi$ algorithm $A$ outputs \texttt{true} (completeness).
We refer to such sound and complete methods and tasks for verification and interpretation collectively using the term \emph{formal reasoning}.

We lay out a framework for the possibilities and challenges of formal reasoning for
transformers by establishing fundamental complexity (and computability) results in this work. 
Thereby, we focus on the so-called \emph{satisfiability ($\empt$) problem} of 
sequence-classifying transformers: 
given a transformer $T$, decide whether there is some input word $w$ such that  $T(w) = 1$, which
can be interpreted as \emph{$T$ accepts $w$}.
Although this may seem like an artificial problem at first glance, it is a natural abstraction of 
problems that commonly occur in almost all non-trivial formal reasoning tasks. Additionally,
since it is detached from the specifics of particular reasoning specifications like safety
properties for instance, uncomputability results and complexity-theoretic hardness results immediately 
transfer to more complex formal reasoning tasks. This also keeps the focus on the transformer 
architecture under consideration. Here, we exclusively consider \emph{transformer encoders} (TE), which are encoder-only.
This is mainly due to the fact that the known high expressive power of encoder-decoder transformers (\cite{PerezBM21}) 
makes formal reasoning trivially impossible. 

Our work is structured as follows. We define necessary preliminaries in Section~\ref{sec:prelim}. 
In Section~\ref{sec:overview}, we give an overview on our complexity results and take a 
comprehensive look at their implications for formal reasoning for transformers. 
In Section~\ref{sec:undec} and Section~\ref{sec:dec} we present our theoretical results: 
we show that $\empt$ is undecidable for classes of TE commonly considered
in research on transformer expressiveness, we show that a bounded version $\bempt$ of the satisfiability problem 
is decidable, for any class of (computable) TE, and give corresponding complexity bounds 
and we show that considering quantized TE, meaning TE 
whose parameters and computations are carried out in a fixed-width arithmetic,
leads to decidability of $\empt$ and give corresponding complexity bounds. Finally, we discuss 
limitations, open problems and future research in Section~\ref{sec:outlook}.


\paragraph{Related work.}
We establish basic computability and complexity results
about transformer-related formal reasoning problems, like formal verification or interpretation. 
This places our work in the intersection between research on \emph{verification and interpretation of transformers}
and \emph{transformer expressiveness}.

There is a limited amount of work concerned with methods for the verification of safety properties of transformers 
(\cite{HsiehCJWHH19,ShiZCHH20,bonaert_fast_2021, DongLJ021}). 
However, all those methods do not fall in the category of formal reasoning, as they are non-complete. 
This means, the rigorous computability and complexity bound established in this work cannot be applied without further considerations. 
The same applies for so far considered interpretability methods (\cite{ZhaoCYLDCWYD24_explainsurvey}). 
We remark that a lot of these approaches are not sound methods either.

In contrast, there is an uprise in theoretical investigations of transformer expressiveness. Initial work dealt with encoder-decoder models 
and showed that such models are Turing-complete (\cite{PerezBM21,bhattamishra-etal-2020-computational}). Note that these are 
different models than the ones we consider, which are encoder-only. Encoder-only models have 
so far been analysed in connection with circuit complexity (\cite{Hahn20, HaoAF22,merrill-etal-2022-saturated,merrill-sabharwal-2023-parallelism}), 
logics (\cite{chiang_tighter_2023,MerrillS23}) and programming languages (\cite{weiss_thinking_2021}). A recently published survey (\cite{10.1162/tacl_a_00663}) 
provides an overview of these results. This work is adjacent as some of the here considered classes of TE, mainly those considered in
Section~\ref{sec:undec}, are motivated by these results and some of the constructions we use in corresponding proofs are similar.

\section{Fundamentals}
\label{sec:prelim}

\paragraph{Mathematical basics.} 
Let $\Sigma$ be a finite set of symbols, called \emph{alphabet}. A \emph{(finite) word $w$ over $\Sigma$} is a finite sequence 
$a_1 \dotsb a_k$ where $a_i \in \Sigma$. We define $|w| = k$. 
As usual, we denote the set of all non-empty words by $\Sigma^+$.  A \emph{language} is a set of words. 
We also extend the notion of an alphabet to vectors $\bs{x}_i \in \reals^d$, meaning that a sequence $\bs{x}_1 \dotsb \bs{x}_k$
is a word over some subset of $\reals^d$.
Usually, we denote vectors using bold symbols like $\bs{x}, \bs{y}$ or $\bs{z}$. 

\paragraph{Transformer encoders (TE).}
We consider the transformer encoders (TE), based on the transformer architecture originally introduced in (\cite{VaswaniSPUJGKP17}).
We take a look at TE from a computability and complexity perspective, making a formal definition of the considered architecture necessary. 
Thereby, we follow the lines of works concerned with formal aspects of transformers like (\cite{Hahn20,PerezBM21,HaoAF22}). From a syntax point of view, 
our definition is most near to (\cite{HaoAF22}).\footnote{As of now, no single definition of Transformer encoders (TE) has been universally adopted in research on their formal aspects, particularly concerning syntax (see the recent survey (\cite{10.1162/tacl_a_00663}) for an overview of different notions of TE). The definition we use here is sufficiently general and provides a parameterized template for the classes of TE considered in our main results.} 

An \emph{TE $T$ with $L$ layers and $h_i$ attention heads in layer $i$} is a tuple 
$(\emb, \{ \att_{i,j} \mid 1 \leq i \leq L, 1 \leq j \leq h_i\}, \{\comb_i \mid 1 \leq i \leq L\}, \out)$ where
\begin{itemize}
	\item $\emb \colon \Sigma \times \nats \rightarrow \reals^{d_0}$ for some $d_0 \in \nats$ is the \emph{positional embedding},
	\item each \emph{attention head} is a tuple  $\att_{i,j} = (\score_{i,j}, \pool_{i,j})$ where
	$\score_{i,j} \colon \reals^{d_{i-1}} \times \reals^{d_{i-1}} \rightarrow \reals$ is a function called \emph{scoring} and 
	$\pool_{i,j} \colon (\reals^{d_{i-1}})^+ \times \reals^+ \rightarrow \reals^{d_{i}}$ is a function called \emph{pooling}, 
	computing  $(\bs{x}_1, \dotsc, \bs{x}_n, s_1, \dotsc, s_n) \mapsto 
	\sum_{i'=1}^n \norm(i', s_1, \dotsc, s_n) (W\bs{x}_{i'})$ 
	where $W$ is a linear map represented by a matrix and $\norm\colon \nats \times \reals^+ \rightarrow \reals$ 
	is a \emph{normalisation},
	\item each $\comb_i \colon \reals^{d_{i,1}}\times \dotsb \times \reals^{d_{i,h_1+1}} \rightarrow \reals^{d_i}$ is called a \emph{combination} 
	and $\out \colon \reals^{d_L} \rightarrow \reals$ is called the \emph{output}.
\end{itemize}
For given $i \leq L$ we call the tuple $(\att_{i,1}, \dotsc, \att_{i,h_i}, \comb_i)$ the \emph{$i$-th layer} of $T$. The TE 
$T$ computes a function $\Sigma^+ \rightarrow \reals$ as follows, also schematically depicted in Figure~\ref{fig:eot_scheme}.
Let $w = a_1, \dotsc, a_n \in \Sigma^+$ be a word. First, $T$ computes an embedding of $w$ by
$\emb(w) = \bs{x}_1^0\dotsb \bs{x}_n^0$ where $\bs{x}^0_i = \emb(a_i, i)$.
Next, each layer $1 \leq i \leq L$ computes a sequence $\bs{x}_1^i \dotsb \bs{x}_n^i$ as follows:
for each input $\bs{x}_m^{i-1}$ and attention head $\att_{i,j}$,
layer $i$ computes $\bs{y}_{m,j}^i = \pool_{i,j}(\bs{x}_1^{i-1}, \dotsc, \bs{x}_n^{i-1}, 
\score_{i,j}(\bs{x}_m^{i-1}, \bs{x}_1^{i-1}), \dotsc, \score_{i,j}(\bs{x}_m^{i-1}, \bs{x}_n^{i-1}))$. 
Then, $\bs{x}^i_m$ is given by $\comb_i(\bs{x}^{i-1}_m, \bs{y}_{m,1}^i, \dotsc, \bs{y}_{m,h_i}^i)$. In the end, the output
$T(w)$ is computed by $\out(\bs{x}^k_n)$, thus the value of the output function for the last symbol of
$w$ after being transformed by the embedding and $L$ layers of $T$. We say that $T$ \emph{accepts}
$w$ if $T(w) = 1$, and we say that $T$ \emph{rejects} $w$ otherwise.  
We call $L$ the \emph{depth} of $T$ and the maximal $h_i$ the \emph{(maximum) width} of $T$. Furthermore,
we call the maximal $d_i$ the \emph{(maximum) dimensionality} of $T$.
Let $\transC,\transC'$ be some classes of TE. We sometimes say that \emph{$\transC'$ is at least as expressive as
$\transC$} or \emph{$\transC$ is at most as expressive as
$\transC'$}, meaning that for each $T \in \transC$ there is $T' \in \transC'$ such that $T$ and $T'$ compute
the same function.
The decision problem $\empt[\transC]$ is given $T \in \transC$ over alphabet $\Sigma$, decide whether there is $w \in \Sigma^+$ such that $T(w) = 1$. We refer to this as the \emph{satisfiability problem} of $\transC$. \footnote{We observe that we can equivalently define \empt{} as requiring $T(w) \geq c$ for arbitrary $c \in \rats$ without changing any of the results detailed in this work.}
\begin{figure}
	\centering
	\input{assets/eot_schematic}
	\caption{Schematic depiction of an TE $T$ with embedding $\emb$ and $k$ (encoder) layers $l_i$. Each layer $l_i$ consists of some $h_i$ attention heads $\att_{i,j}$, 
		whose output is combined by $\comb_i$. Additionally, for some layer $l_i$, the computational flow of $T$ regarding input position $m$ is schematically depicted in detail.}
	\label{fig:eot_scheme}
\end{figure}

\paragraph{Fixed-width arithmetics.}
We consider commonly used \emph{fixed-width arithmetics (FA)} that represent 
numbers using a fixed amount of bits, like floating- or fixed-point arithmetic in this work. See 
(\cite{BaranowskiHLNR20_fpa}) (fixed-point) or (\cite{ConstantinidesD20}) (floating-point) for 
rigorous mathematical definitions of such FA. In this work, however, we only make use of a high-level 
view on different FA. Namely, given some
FA $F$ we assume that all values are represented in binary using 
$b \in \nats$ bits for representing its numbers. 
Thus, there are 
$2^b$ different rational numbers representable in $F$. 
Furthermore, we assume that the considered FA can handle 
overflow situations using either saturation or wrap-around 
and rounding situations by rounding up or off.
We consider TE in the context of $F$.
We say that \emph{$T$ works over $F$}, assuming that all computations 
as well as values occurring in a computation $T(w)$ are carried out in the 
arithmetic defined by $F$.


\section{Overview of complexity results and connection to formal reasoning}
\label{sec:overview}

%

\newtheorem{innercustomthm}{Theorem}
\newenvironment{customthm}[1]
{\renewcommand\theinnercustomthm{#1}\innercustomthm}
{\endinnercustomthm}

We address elementary problems arising in formal reasoning
for transformers in this work. In doing so, we pursue the goal of
establishing basic computability and complexity results for corresponding problems in order to frame possibilities and challenges.

We want our results to be detached from any intricacies of
specific transformer architectures: first, we focus on transformer encoders (TE), so
leaving any decoder mechanism unconsidered. The primary reason for this is that
encoder-decoder architectures are of such high expressive power (\cite{PerezBM21})
that almost all formal reasoning problems are easily seen to be undecidable. The secondary reason for this is that
encoder-decoder architectures subsume encoder-only architectures. So any lower complexity bound, 
established in this work, is also a lower bound for encoder-decoder transformers.

\subsection*{Satisfiability as a baseline formal resoning problem}
To achieve widespread implications of our results, we focus our considerations 
on a fundamental problem arising in formal verification and interpretation 
tasks: given a TE $T$, decide whether there is some input 
$w$ leading to some specific output $T(w)$, as defined formally in terms 
of the \emph{satisfiability problem $\empt[\transC]$} for a class $\transC$ of 
specific TE, see Section~\ref{sec:prelim}. 


To see that this captures the essence of formal reasoning problems occurring in practice, consider the following 
formal verification task: Given a TE $T$, verify that $T$ only accepts inputs where every occurrence of a 
specific key from a set $K$ is accompanied by a particular pattern—for example, a key from $K$ must be 
immediately followed by a value from a set $V$. Such tasks are important to ensure syntactic 
correctness or adherence to some protocol specification. Formally, this is called a robustness property (\cite{ShiZCHH20,Huang_survey}). We can phrase this example task as a 
satisfiability problem by considering the property’s negation, namely, to verify that there exists some input $w$ 
in which a key from $K$ is not properly followed by a value from $V$, yet we have $T(w) = 1$.


Similarly, consider a formal interpretation task where we aim to find the minimal subset $E' \subseteq E$  
of some set of error symbols $E$ such that all inputs $w$ containing all errors in $E'$ are rejected by $T$. For 
instance, in a spam detection system powered by a transformer encoder, $E$ could represent a set of spam 
indicators or malicious keywords. We might want to determine the minimal combination of these indicators that 
will cause the system to classify an input as spam. This is understood as an abductive 
explanation in formal explainable AI (\cite{Silva_formalxai}). Given a candidate subset $E'$, we can certify this by checking that there is some $w$ 
which contains all errors $E'$, but is accepted by $T$. This scenario is a special case of the satisfiability 
problem $\empt[\mathcal{T}]$ for some transformer class $\mathcal{T}$.

\subsection*{Overview of results on the complexity of transformer encoder satisfiability}
We start by considering the class $\transCundec$ of TE, motivated by commonly considered
architectures in the theoretical expressiveness community (\cite{PerezBM21,HaoAF22,Hahn20}): $\transCundec$ consists of
those TE that use a positional embedding, expressive enough to compute a sum, hardmax $\hmax$ as normalisation functions
and a scalar-product based scoring, enriched with a nonlinear map represented by an FNN. 
\begin{customthm}{1}[Section~\ref{sec:undec}]
\label{cthm:empt_undec}
	The satisfiability problem $\empt[\transCundec]$ is undecidable.
\end{customthm}

\begin{figure}[t]
	\centering

\tikzset{every picture/.style={line width=0.75pt}} 

\begin{tikzpicture}[x=0.75pt,y=0.75pt,yscale=-.9,xscale=.9]

\draw   (317,100.71) .. controls (363.94,100.71) and (402,141.11) .. (402,190.94) .. controls (402,240.77) and (363.94,281.17) .. (317,281.17) .. controls (270.06,281.17) and (232,240.77) .. (232,190.94) .. controls (232,141.11) and (270.06,100.71) .. (317,100.71) -- cycle ;
\draw   (317,161.17) .. controls (344.61,161.17) and (367,188.03) .. (367,221.17) .. controls (367,254.3) and (344.61,281.17) .. (317,281.17) .. controls (289.39,281.17) and (267,254.3) .. (267,221.17) .. controls (267,188.03) and (289.39,161.17) .. (317,161.17) -- cycle ;
\draw   (317,51.17) .. controls (386.04,51.17) and (442,102.66) .. (442,166.17) .. controls (442,229.68) and (386.04,281.17) .. (317,281.17) .. controls (247.96,281.17) and (192,229.68) .. (192,166.17) .. controls (192,102.66) and (247.96,51.17) .. (317,51.17) -- cycle ;
\draw    (222,37.17) -- (412,37.17) ;
\draw    (502,161.17) .. controls (455.97,176.51) and (396.54,88.89) .. (329.05,111.01) ;
\draw [shift={(327,111.71)}, rotate = 340.2] [fill={rgb, 255:red, 0; green, 0; blue, 0 }  ][line width=0.08]  [draw opacity=0] (5.36,-2.57) -- (0,0) -- (5.36,2.57) -- cycle    ;
\draw  [color={rgb, 255:red, 0; green, 0; blue, 0 }  ,draw opacity=1 ][fill={rgb, 255:red, 0; green, 0; blue, 0 }  ,fill opacity=1 ] (325.31,111.71) .. controls (325.31,110.78) and (326.07,110.02) .. (327,110.02) .. controls (327.93,110.02) and (328.69,110.78) .. (328.69,111.71) .. controls (328.69,112.65) and (327.93,113.4) .. (327,113.4) .. controls (326.07,113.4) and (325.31,112.65) .. (325.31,111.71) -- cycle ;
\draw  [draw opacity=0] (361.27,116.91) .. controls (357.52,122.8) and (339.01,127.1) .. (316.78,126.87) .. controls (295.09,126.64) and (277.02,122.15) .. (272.8,116.39) -- (316.91,114.26) -- cycle ; \draw   (361.27,116.91) .. controls (357.52,122.8) and (339.01,127.1) .. (316.78,126.87) .. controls (295.09,126.64) and (277.02,122.15) .. (272.8,116.39) ;  
\draw  [draw opacity=0] (354.71,185.79) .. controls (348.57,189.37) and (334.4,191.77) .. (317.93,191.59) .. controls (299.62,191.4) and (284.22,188.08) .. (279.56,183.73) -- (318.04,181.38) -- cycle ; \draw   (354.71,185.79) .. controls (348.57,189.37) and (334.4,191.77) .. (317.93,191.59) .. controls (299.62,191.4) and (284.22,188.08) .. (279.56,183.73) ;  
\draw [draw opacity=0]   (498,136.67) ;
\draw [draw opacity=0]   (97,114.67) ;
\draw [draw opacity=0]   (94,206.17) ;
\draw [draw opacity=0]   (99,201) ;

\draw [draw opacity=0]   (102,175.17) ;
\draw [draw opacity=0]   (100,146.17) ;
\draw    (502,181.17) .. controls (455.97,196.51) and (383.49,155.7) .. (315.74,178.76) ;
\draw [shift={(313.69,179.48)}, rotate = 340.2] [fill={rgb, 255:red, 0; green, 0; blue, 0 }  ][line width=0.08]  [draw opacity=0] (5.36,-2.57) -- (0,0) -- (5.36,2.57) -- cycle    ;
\draw  [color={rgb, 255:red, 0; green, 0; blue, 0 }  ,draw opacity=1 ][fill={rgb, 255:red, 0; green, 0; blue, 0 }  ,fill opacity=1 ] (312,179.48) .. controls (312,178.54) and (312.76,177.79) .. (313.69,177.79) .. controls (314.62,177.79) and (315.38,178.54) .. (315.38,179.48) .. controls (315.38,180.41) and (314.62,181.17) .. (313.69,181.17) .. controls (312.76,181.17) and (312,180.41) .. (312,179.48) -- cycle ;
\draw [line width=0.75]    (152,199) .. controls (179.6,166.74) and (262.43,189.95) .. (280.95,153.48) ;
\draw [shift={(282,151.17)}, rotate = 112.13] [fill={rgb, 255:red, 0; green, 0; blue, 0 }  ][line width=0.08]  [draw opacity=0] (5.36,-2.57) -- (0,0) -- (5.36,2.57) -- cycle    ;
\draw  [color={rgb, 255:red, 0; green, 0; blue, 0 }  ,draw opacity=1 ][fill={rgb, 255:red, 0; green, 0; blue, 0 }  ,fill opacity=1 ] (280.31,151.17) .. controls (280.31,150.23) and (281.07,149.48) .. (282,149.48) .. controls (282.93,149.48) and (283.69,150.23) .. (283.69,151.17) .. controls (283.69,152.1) and (282.93,152.86) .. (282,152.86) .. controls (281.07,152.86) and (280.31,152.1) .. (280.31,151.17) -- cycle ;
\draw    (502,249) .. controls (429.7,249.75) and (464.65,81.67) .. (345.5,82.83) ;
\draw [shift={(343.69,82.86)}, rotate = 358.64] [fill={rgb, 255:red, 0; green, 0; blue, 0 }  ][line width=0.08]  [draw opacity=0] (5.36,-2.57) -- (0,0) -- (5.36,2.57) -- cycle    ;
\draw  [color={rgb, 255:red, 0; green, 0; blue, 0 }  ,draw opacity=1 ][fill={rgb, 255:red, 0; green, 0; blue, 0 }  ,fill opacity=1 ] (252,22.86) .. controls (252,21.92) and (252.76,21.17) .. (253.69,21.17) .. controls (254.62,21.17) and (255.38,21.92) .. (255.38,22.86) .. controls (255.38,23.79) and (254.62,24.55) .. (253.69,24.55) .. controls (252.76,24.55) and (252,23.79) .. (252,22.86) -- cycle ;
\draw [line width=0.75]    (152,109) .. controls (188.61,86.43) and (186.73,22.7) .. (250.72,22.81) ;
\draw [shift={(253.69,22.86)}, rotate = 181.8] [fill={rgb, 255:red, 0; green, 0; blue, 0 }  ][line width=0.08]  [draw opacity=0] (5.36,-2.57) -- (0,0) -- (5.36,2.57) -- cycle    ;
\draw    (502,69) .. controls (456.2,84.27) and (445.16,8.25) .. (384.79,20.54) ;
\draw [shift={(382,21.17)}, rotate = 346.36] [fill={rgb, 255:red, 0; green, 0; blue, 0 }  ][line width=0.08]  [draw opacity=0] (5.36,-2.57) -- (0,0) -- (5.36,2.57) -- cycle    ;
\draw  [color={rgb, 255:red, 0; green, 0; blue, 0 }  ,draw opacity=1 ][fill={rgb, 255:red, 0; green, 0; blue, 0 }  ,fill opacity=1 ] (380.31,21.17) .. controls (380.31,20.23) and (381.07,19.48) .. (382,19.48) .. controls (382.93,19.48) and (383.69,20.23) .. (383.69,21.17) .. controls (383.69,22.1) and (382.93,22.86) .. (382,22.86) .. controls (381.07,22.86) and (380.31,22.1) .. (380.31,21.17) -- cycle ;
\draw [draw opacity=0]   (94,110.5) ;

\draw [draw opacity=0]   (555,251) ;

\draw [draw opacity=0]   (556,70.5) ;

\draw  [color={rgb, 255:red, 0; green, 0; blue, 0 }  ,draw opacity=1 ][fill={rgb, 255:red, 0; green, 0; blue, 0 }  ,fill opacity=1 ] (342,82.86) .. controls (342,81.92) and (342.76,81.17) .. (343.69,81.17) .. controls (344.62,81.17) and (345.38,81.92) .. (345.38,82.86) .. controls (345.38,83.79) and (344.62,84.55) .. (343.69,84.55) .. controls (342.76,84.55) and (342,83.79) .. (342,82.86) -- cycle ;

	\draw (297,71) node [anchor=north west][inner sep=0.75pt]  [font=\tiny] [align=left] {decidable};
	\draw (293,141) node [anchor=north west][inner sep=0.75pt]  [font=\tiny] [align=left] {NEXPTIME};
	\draw (309,203.17) node [anchor=north west][inner sep=0.75pt]  [font=\tiny] [align=left] {NP};
	\draw (293,21) node [anchor=north west][inner sep=0.75pt]  [font=\tiny] [align=left] {undecidable};
	\draw (509,156.67) node [anchor=west] [inner sep=0.75pt]  [font=\tiny]  {$\textsc{btrSat}_{\text{bin}}[\mathcal{T}]$};
	\draw (509.18,176.4) node [anchor=west] [inner sep=0.75pt]  [font=\tiny]  {$\textsc{btrSat}_{\text{un}}[\mathcal{T}]$};
	\draw (570,164) node [anchor=west] [inner sep=0.75pt]  [font=\tiny] [align=left] {(Theorem 3)};
	\draw (94,110.5) node [anchor=west] [inner sep=0.75pt]  [font=\tiny]  {$\textsc{trSat}[\mathcal{T}_{\mathit{udec}}^{\text{log}}]$};
	\draw (97,110.5) node [anchor=east] [inner sep=0.75pt]  [font=\tiny] [align=left] {(Theorem 2)};
	\draw (100,201) node [anchor=west] [inner sep=0.75pt]  [font=\tiny]  {$\textsc{trSat}[\mathcal{T}_{\circ }^{\text{fix}}]$};
	\draw (97,201) node [anchor=east] [inner sep=0.75pt]  [font=\tiny] [align=left] {(Theorem 4)};
	\draw (509,70.5) node [anchor=west] [inner sep=0.75pt]  [font=\tiny]  {$\textsc{trSat}[\mathcal{T}_{\mathit{udec}}]$};
	\draw (570,70.5) node [anchor=west] [inner sep=0.75pt]  [font=\tiny] [align=left] {(Theorem 1)};
	\draw (509,251) node [anchor=west] [inner sep=0.75pt]  [font=\tiny]  {$\textsc{trSat}[\mathcal{T}^{\text{fix}}]$};
	\draw (570,251) node [anchor=west] [inner sep=0.75pt]  [font=\tiny] [align=left] {(Theorem 5)};

\end{tikzpicture}
	\caption{Schematic overview of the computability and complexity results, established in this work. The classes of TE are described in the pretext of the respective theorem. Note that $\transC$ refers to an arbitrary class
	of (computable) TE. The small subset in the classes NP and \NEXPTIME refers to the complete problems. The
	\NEXPTIME-hardness result of $\empt[\transCfix]$ is not visualized}
	\label{fig:overview}
\end{figure}

Essentially, this result implies that even for TE the combination of $\hmax$ normalizations and
expressive scoring is enough to make satisfiability undecidable. Generally, this makes formal reasoning, 
like verifying robustness properties or giving formal explanations, impossible for classes of TE that subsume $\transCundec$.
Specifically, no such methods exist that are fully automatic, sound and complete. Theorem~\ref{cthm:empt_undec}
does not preclude the existence of incomplete methods for instance.

Recently, so-called \emph{log-precision transformers} have been studied (\cite{MerrillS23}).
These transformers are defined as usual, but given a word length $n$ it is assumed that
a log-precision transformer $T$ uses at most $\mathcal{O}(\log(n))$ bits in its internal
computations. To complement these theoretical considerations, we consider the class
$\logtransCundec$ of TE from $\transCundec$ that work with log-precision. Unfortunately, 
this restriction is not enough to circumvent general undecidability.
\begin{customthm}{2}[Section~\ref{sec:undec}]
	The satisfiability problem $\empt[\logtransCundec]$ is undecidable.
\end{customthm}

Given such impossibility results, we turn our attention to the search for decidable cases. 
We make the reasonable assumption that all considered TE are computable, meaning that their 
components like scoring, normalisation, pooling, combination and output functions are computable functions.
Moreover, we assume that each TE $T$ computes its output $T(w)$ for a given input $w$ within polynomial time relative to the size of $T$ and the length of $w$. 
This is reasonable, as the output is computed layer-wise where each layer involves a quadratic amount of calculations per attention head. 
Therefore, the computation depends polynomially on the depth and width of $T$ and the length of $w$.

First, we consider a natural restriction of the satisfiability problem by bounding the length of valid inputs. 
Then satisfiability becomes decidable, regardless of the respective class of TE, but it is difficult from a 
complexity-theoretic perspective. To formalize this, we introduce the \emph{bounded satisfiability problem}
$\bempt[\mathcal{T}]$ for a class $\mathcal{T}$: given an TE $T \in \mathcal{T}$ and a bound $n \in \nats$
on its input length, decide whether there is word $w$ with $|w| \leq n$ s.t.\ $T(w)=1$. 
\begin{customthm}{3}[Section~\ref{sec:dec}, informal restatement]
	The bounded satisfiability problem $\bempt[\transC]$ is decidable for all classes $\transC$ of (computable) TE. 
	Depending on whether $n$ is given in binary or unary coding, $\bempt[\transC]$ is NEXPTIME-, resp.\ NP-complete assuming
	$\transC \supseteq \transCundec$.
\end{customthm}

Informally, this result implies that bounding the word length is a method to enable formal reasoning. However, it does not
change the fact that satisfiability is an essentially hard problem. As hardness 
is a lower bound, this also translates to subsuming formal reasoning tasks.

Imposing a bound on the input length may not be a viable restriction for various formal reasoning tasks. We therefore
study other ways of obtaining decidability. We address the unbounded satisfiability 
problem for practically motivated classes of TE. We consider the class $\periodTransCfix$ of TE that use a positional 
embedding with some periodicity in their positional encoding, 
commonly seen in practice (\cite{VaswaniSPUJGKP17,DufterSS22}), use softmax or hardmax as normalisation 
and which work over some fixed-width arithmetic (FA). 
This last restriction is motivated by recent popular ways to handle ever increasing TE sizes, for example via quantization or using low-bit arithmetics (\cite{BondarenkoNB21}). From a complexity-theoretic perspective, the use of fixed-width arithmetic 
has a similar effect to bounding the input length.  
\begin{customthm}{4}[Section~\ref{sec:dec}] 
	The satisfiability problem $\empt[\periodTransCfix]$ is in \NEXPTIME.
\end{customthm}
So automatic, sound and complete formal reasoning for periodical TE in a fixed-width arithmetic environment is 
generally possible with potentially high complexity. Note that formal reasoning tasks with more complex safety or interpretability 
specifications than simple satisfiability may even lead to higher complexities.

We then aim to show that this is optimal by providing a matching lower bound. However, we need to relax these
restrictions again, namely considering the class $\transCfix$ allowing for TE that use arbitrary
embeddings and work over some fixed-width arithmetic. 
However, due to the fixed-witdth arithmetic assumption, which consequently applies to positional informations as well,
every embedding must necessarily witness a periodic behaviour. Thus, decidability is implied by the same arguments as used in Theorem~4.
Additionally, we show that high complexity is unavoidable, making sound and complete automatic formal reasoning for fixed-width
arithmetic transformers with arbitrary positional embeddings practically intractable.
\begin{customthm}{5}[Section~\ref{sec:dec}]
	The satisfiability problem $\empt[\transCfix]$ is decidable and \NEXPTIME-hard.
\end{customthm}

Figure~\ref{fig:overview} provides a schematic illustration of the computability and complexity results 
summarized in this section. This figure is intended purely for technical clarity, summarizing our findings without 
delving into the formal reasoning implications discussed earlier.

\section{Transformer encoder satisfiability is generally undecidable}
\label{sec:undec}

We consider a class of TE, denoted by $\transCundec$, which we design with the aim of minimising its expressive power, but having an undecidable satisfiability problem.
We define $\transCundec$ by giving minimum requirements:  
positional-embeddings can be of the form $\emb(a_k,0) = (1,1, 0, 0,  k)$ and 
$\emb(a_k,i) = (0,1, i, \sum_{j=0}^{i} j ,  k)$ where we assume some order on the alphabet symbols
$a_1 , a_2, \dotsc$. For scoring functions we allow for $N(\scalar{Q\bs{x}, K\bs{y}})$ where $N$ is a classical Feedforward Neural Network
(FNN) with $\relu$ activations, $Q$ and $K$ are linear maps and $\scalar{\dotsb}$ denotes the usual scalar product, 
for normalisations we allow for hardmax $\hmax(i, x_1, \dotsc, x_n) = \frac{1}{m}$ if 
$x_i \geq x_j$ for all $j \leq n$ and there are $m$ distinct $x_j$ such that $x_i = x_j$ otherwise $\hmax(i, x_1, \dotsc, x_n) = 0$.
Combinations as well as output functions can be classical FNN with $\relu$ activation. Aside from technical reasons,
we motivate the choice of $\transCundec$ in Section~\ref{sec:overview}.
To ease our notation, we exploit the fact that using $\hmax$ as normalisation implies a clearly
defined subset of positions $M$ that are effective in the computation of some attention head $\att$ 
given some position $i$, namely those that are weighted non-zero. In this case, we say that $\att$ \emph{attends to $M$ given position $i$}.

\begin{figure}[t]
	\centering

\tikzset{every picture/.style={line width=0.75pt}} 

\begin{tikzpicture}[x=0.75pt,y=0.75pt,yscale=-1,xscale=1]
	
	\draw  [draw opacity=0][fill={rgb, 255:red, 0; green, 0; blue, 0 }  ,fill opacity=1 ] (446.84,47.82) .. controls (446.84,46.83) and (446.04,46.03) .. (445.05,46.03) .. controls (444.06,46.03) and (443.26,46.83) .. (443.26,47.82) .. controls (443.26,48.8) and (444.06,49.6) .. (445.05,49.6) .. controls (446.04,49.6) and (446.84,48.8) .. (446.84,47.82) -- cycle ;
	\draw  [draw opacity=0][fill={rgb, 255:red, 0; green, 0; blue, 0 }  ,fill opacity=1 ] (446.84,39.44) .. controls (446.84,38.45) and (446.04,37.65) .. (445.05,37.65) .. controls (444.06,37.65) and (443.26,38.45) .. (443.26,39.44) .. controls (443.26,40.43) and (444.06,41.23) .. (445.05,41.23) .. controls (446.04,41.23) and (446.84,40.43) .. (446.84,39.44) -- cycle ;
	\draw  [draw opacity=0][fill={rgb, 255:red, 0; green, 0; blue, 0 }  ,fill opacity=1 ] (446.84,31.06) .. controls (446.84,30.08) and (446.04,29.28) .. (445.05,29.28) .. controls (444.06,29.28) and (443.26,30.08) .. (443.26,31.06) .. controls (443.26,32.05) and (444.06,32.85) .. (445.05,32.85) .. controls (446.04,32.85) and (446.84,32.05) .. (446.84,31.06) -- cycle ;
	\draw  [draw opacity=0][fill={rgb, 255:red, 0; green, 0; blue, 0 }  ,fill opacity=1 ] (269.67,116.84) .. controls (270.66,116.84) and (271.46,116.04) .. (271.46,115.05) .. controls (271.46,114.07) and (270.66,113.27) .. (269.67,113.27) .. controls (268.69,113.27) and (267.89,114.07) .. (267.89,115.05) .. controls (267.89,116.04) and (268.69,116.84) .. (269.67,116.84) -- cycle ;
	\draw  [draw opacity=0][fill={rgb, 255:red, 0; green, 0; blue, 0 }  ,fill opacity=1 ] (278.05,116.84) .. controls (279.03,116.84) and (279.83,116.04) .. (279.83,115.05) .. controls (279.83,114.07) and (279.03,113.27) .. (278.05,113.27) .. controls (277.06,113.27) and (276.26,114.07) .. (276.26,115.05) .. controls (276.26,116.04) and (277.06,116.84) .. (278.05,116.84) -- cycle ;
	\draw  [draw opacity=0][fill={rgb, 255:red, 0; green, 0; blue, 0 }  ,fill opacity=1 ] (286.42,116.84) .. controls (287.41,116.84) and (288.21,116.04) .. (288.21,115.05) .. controls (288.21,114.07) and (287.41,113.27) .. (286.42,113.27) .. controls (285.44,113.27) and (284.64,114.07) .. (284.64,115.05) .. controls (284.64,116.04) and (285.44,116.84) .. (286.42,116.84) -- cycle ;
	
	\draw [color={rgb, 255:red, 0; green, 0; blue, 0 }  ,draw opacity=1 ]   (122.82,131.36) .. controls (144.62,166.57) and (328.61,194.39) .. (391.38,171.21) ;
	\draw [shift={(394.16,170.12)}, rotate = 157.26] [fill={rgb, 255:red, 0; green, 0; blue, 0 }  ,fill opacity=1 ][line width=0.08]  [draw opacity=0] (5.36,-2.57) -- (0,0) -- (5.36,2.57) -- cycle    ;
	\draw [color={rgb, 255:red, 0; green, 0; blue, 0 }  ,draw opacity=1 ]   (148.66,131.36) .. controls (208.42,162.63) and (341.49,163.81) .. (391.92,132.8) ;
	\draw [shift={(394.16,131.36)}, rotate = 146.24] [fill={rgb, 255:red, 0; green, 0; blue, 0 }  ,fill opacity=1 ][line width=0.08]  [draw opacity=0] (5.36,-2.57) -- (0,0) -- (5.36,2.57) -- cycle    ;
	\draw [color={rgb, 255:red, 0; green, 0; blue, 0 }  ,draw opacity=1 ]   (174.5,131.36) .. controls (222.64,155.87) and (359.26,153.93) .. (430.78,132.03) ;
	\draw [shift={(432.93,131.36)}, rotate = 162.37] [fill={rgb, 255:red, 0; green, 0; blue, 0 }  ,fill opacity=1 ][line width=0.08]  [draw opacity=0] (5.36,-2.57) -- (0,0) -- (5.36,2.57) -- cycle    ;
	\draw [color={rgb, 255:red, 0; green, 0; blue, 0 }  ,draw opacity=1 ]   (226.19,105.51) .. controls (275.36,59.9) and (374.51,45.14) .. (430.4,65.78) ;
	\draw [shift={(432.93,66.75)}, rotate = 201.71] [fill={rgb, 255:red, 0; green, 0; blue, 0 }  ,fill opacity=1 ][line width=0.08]  [draw opacity=0] (5.36,-2.57) -- (0,0) -- (5.36,2.57) -- cycle    ;
	\draw [color={rgb, 255:red, 0; green, 0; blue, 0 }  ,draw opacity=1 ]   (200.35,105.51) .. controls (251.26,67.33) and (343.08,35.34) .. (391.96,65.34) ;
	\draw [shift={(394.16,66.75)}, rotate = 213.85] [fill={rgb, 255:red, 0; green, 0; blue, 0 }  ,fill opacity=1 ][line width=0.08]  [draw opacity=0] (5.36,-2.57) -- (0,0) -- (5.36,2.57) -- cycle    ;
	\draw [color={rgb, 255:red, 0; green, 0; blue, 0 }  ,draw opacity=1 ]   (252.03,105.51) .. controls (294.42,69.27) and (381.59,35.05) .. (469.04,65.8) ;
	\draw [shift={(471.69,66.75)}, rotate = 200.13] [fill={rgb, 255:red, 0; green, 0; blue, 0 }  ,fill opacity=1 ][line width=0.08]  [draw opacity=0] (5.36,-2.57) -- (0,0) -- (5.36,2.57) -- cycle    ;
	\draw   (471.69,66.75) -- (497.53,66.75) -- (497.53,92.59) -- (471.69,92.59) -- cycle ;
	
	\draw   (432.93,66.75) -- (458.77,66.75) -- (458.77,92.59) -- (432.93,92.59) -- cycle ;
	\draw   (394.16,66.75) -- (420.01,66.75) -- (420.01,92.59) -- (394.16,92.59) -- cycle ;
	\draw   (432.93,105.51) -- (458.77,105.51) -- (458.77,131.36) -- (432.93,131.36) -- cycle ;
	\draw   (394.16,105.51) -- (420.01,105.51) -- (420.01,131.36) -- (394.16,131.36) -- cycle ;
	\draw   (394.16,144.28) -- (420.01,144.28) -- (420.01,170.12) -- (394.16,170.12) -- cycle ;
	\draw [color={rgb, 255:red, 0; green, 0; blue, 0 }  ,draw opacity=1 ][line width=2.25]    (420.01,79.67) -- (432.93,79.67) ;
	\draw [color={rgb, 255:red, 0; green, 0; blue, 0 }  ,draw opacity=1 ][line width=2.25]    (458.77,79.67) -- (471.69,79.67) ;
	\draw [color={rgb, 255:red, 0; green, 0; blue, 0 }  ,draw opacity=1 ][line width=2.25]    (420.01,118.43) -- (432.93,118.43) ;
	\draw [color={rgb, 255:red, 0; green, 0; blue, 0 }  ,draw opacity=1 ][line width=2.25]    (407.09,144.28) -- (407.09,131.36) ;
	\draw [color={rgb, 255:red, 0; green, 0; blue, 0 }  ,draw opacity=1 ][line width=2.25]    (407.09,92.59) -- (407.09,105.51) ;
	\draw [color={rgb, 255:red, 0; green, 0; blue, 0 }  ,draw opacity=1 ]   (316.64,105.51) .. controls (325.33,97.62) and (351.17,97.71) .. (362.67,104.29) ;
	\draw [shift={(365,105.89)}, rotate = 219.52] [fill={rgb, 255:red, 0; green, 0; blue, 0 }  ,fill opacity=1 ][line width=0.08]  [draw opacity=0] (3.57,-1.72) -- (0,0) -- (3.57,1.72) -- cycle    ;
	\draw [color={rgb, 255:red, 0; green, 0; blue, 0 }  ,draw opacity=1 ][line width=2.25]    (445.85,92.59) -- (445.85,105.51) ;
	\draw [color={rgb, 255:red, 0; green, 0; blue, 0 }  ,draw opacity=1 ][line width=2.25]    (503.99,118.43) -- (503.99,131.36) ;
	\draw [color={rgb, 255:red, 0; green, 0; blue, 0 }  ,draw opacity=1 ][line width=2.25]    (497.53,124.89) -- (510.46,124.89) ;
	
	\draw (389.83,160.73) node [anchor=east] [inner sep=0.75pt]  [font=\tiny]  {$0\ $};
	\draw (389.83,121.96) node [anchor=east] [inner sep=0.75pt]  [font=\tiny]  {$1\ $};
	\draw (389.83,83.2) node [anchor=east] [inner sep=0.75pt]  [font=\tiny]  {$2\ $};
	\draw (404.5,177.51) node [anchor=north] [inner sep=0.75pt]  [font=\tiny]  {$0$};
	\draw (443.26,177.51) node [anchor=north] [inner sep=0.75pt]  [font=\tiny]  {$1$};
	\draw (482.03,177.51) node [anchor=north] [inner sep=0.75pt]  [font=\tiny]  {$2$};
	\draw (112.74,111.52) node [anchor=north west][inner sep=0.75pt]  [font=\small] [align=left] {a};
	\draw (138.58,111.52) node [anchor=north west][inner sep=0.75pt]  [font=\small] [align=left] {a};
	\draw (164.42,108.94) node [anchor=north west][inner sep=0.75pt]  [font=\small] [align=left] {b};
	\draw (190.27,108.94) node [anchor=north west][inner sep=0.75pt]  [font=\small] [align=left] {d};
	\draw (216.11,111.52) node [anchor=north west][inner sep=0.75pt]  [font=\small] [align=left] {c};
	\draw (241.95,108.94) node [anchor=north west][inner sep=0.75pt]  [font=\small] [align=left] {d};
	\draw (321.94,109.67) node [anchor=north west][inner sep=0.75pt]  [font=\tiny,color={rgb, 255:red, 0; green, 0; blue, 0 }  ,opacity=1 ] [align=left] {Lemma 1};
	\draw (503.99,136.2) node [anchor=north] [inner sep=0.75pt]  [font=\tiny,color={rgb, 255:red, 0; green, 0; blue, 0 }  ,opacity=1 ] [align=left] {Lemma 2};
	\draw (484.61,79.67) node  [font=\small] [align=left] {d};
	\draw (445.85,79.67) node  [font=\small] [align=left] {c};
	\draw (407.09,79.67) node  [font=\small] [align=left] {d};
	\draw (445.85,118.43) node  [font=\small] [align=left] {b};
	\draw (407.09,118.43) node  [font=\small] [align=left] {a};
	\draw (407.09,157.2) node  [font=\small] [align=left] {a};

\end{tikzpicture}
	\caption{Schematic depiction of the expressive capabilities of TE from $\transCundec$ in context of 
	the \unbOTWP proven in Lemma~\ref{sec:undec;lem:decode} and Lemma~\ref{sec:undec;lem:lin}.}
	\label{sec:undec;fig:tiling_to_word}
\end{figure}

We prove that $\empt[\transCundec]$ is undecidable by establishing a reduction from the \emph{(unbounded) octant tiling-word problem (\unbOTWP)}. 
For details on tiling problems, see Appendix~\ref{app:tiling}. The $\unbOTWP$ is defined as follows: given a tiling system $\mathcal{S} = (S,H,V,t_I,t_F)$ 
where $S$ is some finite set of tiles, $H, V \subseteq S^2$ and $t_I, t_F \in S$ we have to decide whether there is a word (a) 
$t_{0,0},t_{1,0},t_{1,1},t_{2,0},t_{2,1},t_{2,2}, t_{3,0},\ldots,$ $t_{k,k} \in S^+$ such
that (b) $t_{0,0} = t_I$, $t_{k,k} = t_F$, (c) for all $i \leq k$ and $0 \leq j < i$ holds $(t_{i,j}, t_{i,j+1}) \in H$ and (d) for all $i \leq k-1$ and $j \leq i$ holds 
$(t_{i,j}, t_{i+1,j}) \in V$. We call a word $w$ which satisfies (a) an \emph{encoded tiling} and if (b)-(d) are satisfied as well then 
we call $w$ a \emph{valid} encoded tiling. 
Our proof strategy is easily described: given a tiling system $\mathcal{S}$, we build an TE 
$T_\mathcal{S} \in \transCundec$ which accepts a word $w$ if it fulfils conditions (a) to (d) and otherwise $T_\mathcal{S}$ rejects $w$.
We derive most technical proofs of the following lemmas and theorems to Appendix~\ref{app:undec_proofs} and instead provide intuitions 
and proof sketches in this section. 

We start with the first observation: the expressiveness of TE in $\transCundec$ is sufficient to decode the octant tiling potentially represented by a given word $w$, as depicted by the arrows in Figure~\ref{sec:undec;fig:tiling_to_word}.
In detail, two encoder layers in combination with a positional embedding definable in $\transCundec$ are expressive enough to compute for a given symbol $t$ in $w$ to which position in an octant tiling it corresponds, if we interpret $w$ as an encoded tiling.
\begin{lemma}
	\label{sec:undec;lem:decode}
 	Let $\mathcal{S}$ be a tiling system with tiles $S = \{a_1, \dotsc, a_k\}$.  
 	There is an embedding function $\emb$ and there are encoder layers 
 	$l_1$ and $l_2$ definable in $\transCundec$ such that for each 
 	word $w = t_{0,0} t_{1,0} t_{1,1} t_{2,0} \dotsb t_{m,n} \in S^+$ holds that
 	$l_2(l_1(\emb(w))) = \bs{x}^2_1 \dotsc \bs{x}^2_{|w|}$ where $\bs{x}^2_i=(1, i, r(i), c(i), k_i)$ such that 
 	$a_{k_i}$ is equal to the symbol at position $i$ in $w$ and 
 	$(r(1), c(1)), (r(2), c(2)), \dotsc, (r(|w|), c(|w|))$ is equal to $(0,0), (1,0) \dotsc, (m,n)$. 
\end{lemma}

Assume that $w \in S^+$. Lemma~\ref{sec:undec;lem:decode} implies that a
TE $T \in \transCundec$ is generally able to recognize whether $w$ is an encoded tiling 
as soon as $T$ is able to check whether $r(|w|)$ and $c(|w|)$ of the last symbol of $w$ processed by $l_2(l_1(\emb(\dotsb)))$ are equal.
Therefore, property (a) and also (b) can be checked by TE in $\transCundec$ using the residual connection and FNN.

Property (c) can be ensured if it is possible to build an attention head that is able to attend to position $k+1$ given
position $k$. Let $w = t_{0,0} t_{1,0} t_{1,1} t_{2,0} \dotsb t_{m,m}$ 
with $t_{i,j} \in S$. To verify whether property (d) holds, a TE must be able to attend
to position $k + (i+1)$ given position $k$ corresponding to symbol $t_{i,j}$. This is depicted in Figure~\ref{sec:undec;fig:tiling_to_word} 
by the bold lines between horizontal and vertical tiles.
In summary, to check properties (a) -- (d) it is left to argue that there are attention heads in $\transCundec$ that can attend to positions 
depending linearly on the values of the currently considered position. 
\begin{lemma}
	\label{sec:undec;lem:lin}
	Let $f(x_1, \dotsc, x_k) = a_1 x_1 + \dotsb + a_k x_k + b$ with $a_i, b \in \reals$ be some linear function. 
	There is attention head $\att_f$ in $\transCundec$ 
	such that for all sequences $\bs{x}_1, \dotsc, \bs{x}_m$ where all $\bs{x}_i = (1,i, \bs{y}_i)$ for some $\bs{y}_i \in \reals^{k-2}$
	attention head $\att_f$ attends to $\{\bs{x}_{j}, \bs{x}_{j+1}\}$ given position $i$ if $f(\bs{x}_i) = j + \frac{1}{2}$ with $j \leq m-1$
	and otherwise to $\{\bs{x}_j\}$ where $j$ is the value nearest to $f(\bs{x}_i)$.
\end{lemma}

In combination, the previous lemmas indicate that TE from $\transCundec$ are able to verify whether a given word is a valid encoded tiling. 
This expressive power is enough, to lead to an undecidable satisfiability problem for TE from $\transCundec$.
\begin{theorem}
	\label{sec:undec;thm:undec}
	The decision problem $\empt[\transCundec]$ is undecidable.
\end{theorem}
\begin{proof}[Proof Sketch]
	We establish a reduction from $\unbOTWP$ to $\empt[\transCundec]$ by constructing
	for each instance $\mathcal{S} = (S, H, V, t_I, t_F)$ of $\unbOTWP$ an TE $T_{\mathcal{S}}$ accepting
	exactly those $w$ corresponding to a valid encoded-tiling for $\mathcal{S}$.
	
	$T_{\mathcal{S}}$ uses the positional embedding described in the beginning of Section~\ref{sec:undec} and has four layers. Layers $l_1$ and $l_2$ are given by Lemma~\ref{sec:undec;lem:decode} and are used to decode the row and column
	indexes corresponding to a potential octant tiling for each symbol in a given word $w$.
	Layer $l_3$ uses the informations encoded by the embedding and the decoded
	row and column indexes to check whether properties (a) to (d) described above hold
	for $w$. The necessary informations are aggregated using three attention heads 
	$\att_{\mathit{prev}}$, $\att_{\mathit{next}}$ and $\att_{\mathit{step}}$, each built 
	according to Lemma~\ref{sec:undec;lem:lin}.Thereby, $\att_{\mathit{prev}}$ attends
	each position to its predecessor, but the first position attends to itself. 
	This allows to clearly identify the vector corresponding to the first position in $w$ and
	check whether this is equal to tile $t_I$.
	Attention head $\att_{\mathit{next}}$ attends each position to its successor, but the last 
	position attends to itself. This allows to clearly identify the vector corresponding to the last 
	position in $w$, in order to check whether this is equal to $t_F$, and to check conditions 
	given by $H$. Attention head $\att_{\mathit{step}}$ attends each position to the position
	with the same column index but the successive row index. If there is no such successive row it attends to the last position. This allows to check whether conditions given by $V$
	holds. Each of these conditions is checked in the combination function of $l_3$, using
	specifically built feed-forward neural networks outputting $0$ to some predefined vector dimension if and only if the condition is met. Finally, layer $l_4$ aggregates the information 
	of all positions in the vector corresponding to the last position using attention head $\att_{\text{leq}}$, again given by Lemma~\ref{sec:undec;lem:lin}. 
	The correctness of this reduction follows from the detailed construction of $T_{\mathcal{S}}$, given in Appendix~\ref{app:undec_proofs}.
\end{proof}

Next, we consider the class $\logtransCundec$ which is defined exactly like $\transCundec$ 
but for all $T \in \logtransCundec$ working over alphabet $\Sigma$ and all words 
$w$ with $|w| = n$ we assume that $T(w)$ is carried out in some fixed-width arithmetic 
$F$ using $\mathcal{O}(\log(\max(|\Sigma|, n)))$ bits.
\begin{theorem}
	\label{sec:undec;thm:logundec}
	The decision problem $\empt[\logtransCundec]$ is undecidable.
\end{theorem}
\begin{proof}[Proof sketch]
	This proof follows the exact same line as the proof of Theorem~\ref{sec:undec;thm:undec}.
	Additionally, we need to argue that $T_{\mathcal{S}}$ works as intended, despite 
	the fact that it is limited by some log-precision $F$. 
	
	Looking at the proof of Theorem~\ref{sec:undec;thm:undec}, it is imminent 
	that the magnitude and precision of all values used and produced in the computation $T_{\mathcal{S}}(w)$ depend polynomially on $n$ and, thus, we can choose the representation of $F$
	to be linear in $\log(n)$, which avoids any overflow or rounding situations and ensures
	that $T_{\mathcal{S}}$ works as intended.	
	A formal proof is given in Appendix~\ref{app:undec_proofs}.
\end{proof}

\section{How to make transformer encoder satisfiability decidable}
\label{sec:dec}

In this section we investigate classes of TE leading to decidable $\empt$ problems or decidable restrictions of it.
Additionally, we establish corresponding complexity bounds.

In order to establish clearly delineated upper complexity bounds, we need to bound the  
representation size of a TE $T$. Instead of tediously analyzing the space needed to represent embedding, scoring, pooling, combination 
and normalisation functions, we note that it suffices to estimate the size up to polynomials only. The \emph{complexity} of a TE $T$ 
with $L$ layers and $h_i$ attention heads in layer $i$, working on inputs over alphabet $\Sigma$, is $|T| := |\Sigma| + L + H + D$ where 
$H := \max \{ h_i \mid 1 \le i \le L \}$ and $D$ is the maximal dimensionality of vectors occurring in a computation of $T$. Note that one
can reasonably assume the \emph{size} of a syntactic representation of $T$ to be polynomial in $|T|$, and that TE have the
\emph{polynomial evaluation property}: given a word $w \in \Sigma^+$, $T(w)$ can be computed in time that is polynomial in $|T| + |w|$.
Section~\ref{sec:overview} discusses why this assumption is reasonable.

We start with a natural restriction: bounding the word length. Let $\transC$ be a class of TE. The \emph{bounded satisfiability problem}, denoted by $\bempt[\transC]$ is: given $T \in \transC$ 
and some $n \in \nats$, decide whether there is a word $w$ with $|w| \leq n$ such that $T(w) = 1$. It is not hard
to see that $\bempt[\transC]$ is decidable. However, its complexity depends on the value of $n$, and we therefore
distinguish whether $n$ is represented in \emph{binary} or \emph{unary} encoding. We denote the problems as
$\bempt_{\mathsf{bin}}[\transC]$ and $\bempt_{\mathsf{un}}[\transC]$.
 
\begin{theorem}
	\label{sec:dec;thm:boundedword}
	Let $\transC$ be a class of TE. Then
	\begin{enumerate}
		\item $\bempt_{\mathsf{un}}[\transC]$ is in NP; if 
		$\transCundec \subseteq \transC$ then $\bempt_{\mathsf{un}}[\transC]$ is NP-complete,
		\label{sec:dec;thm:boundedword;prop:NP}
		\item $\bempt_{\mathsf{bin}}[\transC]$ is in NEXPTIME;
		if $\transCundec \subseteq \transC$ then $\bempt_{\mathsf{bin}}[\transC]$ is NEXPTIME-complete.
		\label{sec:dec;thm:boundedword;prop:NEXPTIME}
	\end{enumerate}
\end{theorem}

\begin{proof}[Proof Sketch]
The decidability result of statement (\ref{sec:dec;thm:boundedword;prop:NP}) can be shown using a simple guess-and-check argument: given $n \in \nats$, 
guess a word $w \in \Sigma^+$ with $|w| \le n$, compute $T(w)$ and check that the result is $1$. This is possible in time polynomial
in $|T| + n$ using the polynomial evaluation property. Note that $|T|$ depends on $|\Sigma|$, thus this also respects the actual representation size of $w$. Moreover, the value of $|T| + n$ is polynomial in the size needed to represent
$n$ in unary encoding. 
The decidability result of statement (\ref{sec:dec;thm:boundedword;prop:NEXPTIME}) is shown along the same lines. However, if the value $n$ is encoded binarily
then this part of the input is of size $\log n$, and $|T|+n$ becomes exponential in this. Hence, the guess-and-check procedure only proves
that $\bempt_{\mathsf{bin}}[\transC] \in$ NEXPTIME.

For the completeness result in (\ref{sec:dec;thm:boundedword;prop:NP})
it suffices to argue that the problem is NP-hard. We use that TE in $\transCundec$ 
are expressive enough to accept a given word $w$ if and only if it is a valid encoded tiling, cf.\ Section~\ref{sec:undec} for details.
It is possible to establish NP-hardness of a restriction of the octant word-tiling problem, namely the 
\emph{bounded octant word-tiling problem} (for unarily encoded input values). See Appendix~\ref{app:tiling} for details on tiling problems. 
It remains to observe that the construction
in Theorem~\ref{sec:undec;thm:undec} is in fact a polynomial-time reduction, and that it reduces the bounded octant word-tiling 
problem to the bounded satisfiability problem. The argument for NEXPTIME-hardness in statement (\ref{sec:dec;thm:boundedword;prop:NEXPTIME})
works the same with, again, the bounded octant-word tiling problem shown to be NEXPTIME-hard when the input parameter $n$
is encoded binary.
A full proof for Theorem~\ref{sec:dec;thm:boundedword} is in Appendix~\ref{app:dec_proofs}.
\end{proof}


We turn our attention to classes of TE that naturally arise in practical contexts. We consider TE that work over some 
fixed-width arithmetic, like fixed- or floating-point numbers, and which have an embedding relying on a periodical encoding of positions.
We start with establishing a scenario where $\empt$ is decidable in \NEXPTIME. Regardless of the underlying TE class $\transC$,
our proof strategy always relies on a certifier-based understanding of \NEXPTIME: given $T \in \transC$, we nondeterministically
guess a word $w$, followed by a deterministic certification whether $T(w) = 1$ holds. For this to show $\empt[\transC] \in \NEXPTIME$, 
we need to argue that the overall running time of such a procedure is at most exponential, in particular that whenever there is a 
word $w$ with $T(w)=1$ then there is also some $w'$ with $T(w') = 1$ and $|w'| \leq 2^{\poly(|T|)}$. Again, we rely on the polynomial
evaluation property of TE in $\transC$, i.e.\ the fact that $T(w')$ can be computed in time polynomial in $|T| + |w'|$. 

We consider the class of TE $\periodTransCfix$, defined by placing restrictions on the positional embedding of an TE
$T$ to be \emph{additive-periodical} which means that $\emb(a,i) = \emb'(a) + \pos(i)$ where $\pos$ is periodical, i.e.\ there is 
$p \geq 1$ such that $\pos(i) = \pos(i+p)$ for all $i \in \nats$. Additionally, all normalisation functions are realised by 
either the softmax function $\softmax$ or the hardmax function $\hmax$. Moreover, we assume that all computations occurring in $T$ 
are carried out in some fixed-width arithmetic, encoding values in binary using a fixed number $b \in \nats$ of bits.
Aside from technical reasons, we motivate the choice of $\periodTransCfix$ in Section~\ref{sec:overview}. Given these restrictions,
we adjust the definition of the complexity of $T \in \periodTransCfix$ as a measure of the size (up to polynomials) as 
$|T| := |\Sigma| + L + H + D + p + b$.

\begin{lemma}
	\label{sec:decidability;lem:shortword_softmax}
	There is a polynomial function $\poly\colon \nats \to \nats$ such that for all $T \in \periodTransCfix$ and all words $w$ with $T(w) = 1$ there 
	is word $w'$ with $T(w') = 1$ and $|w'| \leq 2^{\poly(|T|)}$.  
\end{lemma}
\begin{proof} [Proof Sketch]
The polynomial $\poly$ can be chosen uniformly for all $T \in \periodTransCfix$ because 
for all positional embeddings of TE in $\periodTransCfix$ there is an upper bound on the
period and on the bit-width in the underlying arithmetic.
The small-word property stated by the lemma is then shown by arguing, given polynomial $\poly$, TE $T$ and $|w| > 2^{\poly(|T|)}$,  
that $w$ contains unnecessary subwords $u$ that can be cut out without changing the output in $T$. Here, we exploit the fact $T$ has some periodicity $p$ and only consider those $u$ whose length is a multitude of $p$. This ensures that the resulting word
$w'$, given by $w$ without $u$, is embedded the same way as $w$ by the positional
embedding of $T$.
The existence of such subwords follows from $T$'s limited distinguishing capabilities, 
especially in its normalisations, due to the bounded representation size of numerical values
possible in the underlying fixed-width arithmetic.
A formal proof relies
on basic combinatorial arguments and given in Appendix~\ref{app:dec_proofs}. 
\end{proof}


\begin{theorem}
	\label{sec:decidability;thm:nexptime}
	$\empt[\periodTransCfix]$ for additive-periodical TE over fixed-width arithmetic is in \NEXPTIME.
\end{theorem}
\begin{proof}
    Let $T \in \periodTransCfix$ working over alphabet $\Sigma$.
	We use a certifier-based understanding of a nondeterministic exponential-time algorithm as follows:
	We (a) guess an input $w \in \Sigma^+$ and (b) compute $T(w)$ to check whether $T(w)=1$. 
	For correctness, we need to argue that the length of $w$ 
	is at most exponential in $|T|$. This argument is given by 
	Lemma~\ref{sec:decidability;lem:shortword_softmax}.
	Note that via assumption we have that $T(w)$ can be computed in 
	polynomial time regarding $|T|$ and $|w|$. 
\end{proof}

Next, we address the goal of obtaining a matching lower bound, i.e.\ \NEXPTIME-hardness. An obvious way to do so would be to follow 
Theorem~\ref{sec:dec;thm:boundedword}.\ref{sec:dec;thm:boundedword;prop:NEXPTIME} and form a reduction from the 
bounded octant word-tiling problem.
Hence, given a tiling system $\mathcal{S}$ and $n \in \nats$ encoded binarily, we would have to construct -- in time
polynomial in $|\mathcal{S}| + \log n$ -- an TE $T_{\mathcal{S},n} \in \periodTransCfix$ such that $T_{\mathcal{S},n}(w) = 1$
for some $w \in \Sigma^+$ iff there is a word $w = t_{1,1},t_{2,1},t_{2,2},t_{3,1},\ldots,t_{n,n}$ representing a valid 
$\mathcal{S}$-tiling. In particular, $T_{\mathcal{S},n}$ would have to be able to recognise the correct word length and 
reject input that is longer than $|w| = \frac{n(n+1)}{2}$. 
This poses a problem for TE with periodical embeddings. To recognize whether a word is too long, 
an TE $T$ must ultimately rely on its positional embedding, which seems to make a periodicity of $p \geq \frac{n(n+1)}{2}$ necessary. 
Since the size of periodical TE is linear in $p$, we get an exponential blow-up in a potential reduction 
of $\bOTWP_{\mathsf{bin}}$ to $\empt[\periodTransCfix]$, given that the values of $\frac{n(n+1)}{2}$ and already $n$ are
exponential in the size of a binary representation of $n$. This problem vanishes when the requirement of the underyling
positional embedding to be periodical is lifted: allowing for arbitrary TE, working over some fixed-width arithmetic, leads
to an \NEXPTIME-hard satisfiability problem. Let $\transC^\textsc{fix}$ be defined similar to $\periodTransCfix$, but we allow for 
arbitrary embeddings. Furthermore, we assume that the considered fixed-width arithmetics can handle overflow situations
using saturation.
\begin{theorem}
	\label{sec:decidability;thm:nexptimehard}
	$\empt[\transCfix]$ for TE over fixed-width arithmetic is decidable and \NEXPTIME-hard.
\end{theorem}
\begin{proof}[Proof sketch]
	The decidability follows from the same arguments as in Theorem~\ref{sec:decidability;thm:nexptime}, 
	with the insight that even though $T \in \transCfix$ uses an arbitrary embedding, a fixed-width setting
	with $b$ bits enforces a periodicity of size at most $2^b$, assuming overflow is handled by wrap-around, or a 
	periodicity of size $1$ after a finite prefix of length up to $2^b$, assuming overflow is handled by saturation, 
	in the embedding.
	For the hardness, we establish a reduction from $\bOTP_{\mathsf{bin}}$ to $\empt[\transCfix]$ by 
	constructing, for each instance $(\mathcal{S},n)$ of $\bOTP_{\mathsf{bin}}$, an TE 
	$T_{\mathcal{S},n}$ working over some fixed-width arithmetic $F$, which accepts exactly those 
	$w$ with $|w| = \frac{n(n+1)}{2}$ corresponding to a valid word-encoded tiling for $\mathcal{S}$. 
	See Appendix~\ref{app:tiling} for details on tiling problems.
	The construction is similar to the one given for $T_\mathcal{S}$ in the proof of 
	Theorem~\ref{sec:undec}, but we need to enable $T_{\mathcal{S},n}$ to reject words 
	that are too long corresponding a polynomial bound dependent on $n$. 
	This implies that $T_{\mathcal{S},n}$, based on the positional embedding $\emb$ 
	specified in Section~\ref{sec:undec}, is able to check for all symbols if their respective 
	position is less than or equal to a predefined bound. This can be achieved with similar tools as used in Lemma~\ref{sec:undec;lem:lin}.
	Furthermore, we need to ensure that $T_{\mathcal{S},n}$ works as intended, despite the
	fact that it is limited by $F$. The arguments follow the same line as the proof of Theorem~\ref{sec:undec;thm:logundec}.
	A formal proof is given in Appendix~\ref{app:dec_proofs}.
\end{proof}

\section{Summary, limitations and outlook}
\label{sec:outlook}

We investigated the satisfiability problem of transformer encoders (TE). 
In particular, we considered the computational complexity of the satisfiability problem $\empt$ 
of TE in context of different classes of TE, forming a baseline for understanding possibilities and challenges 
of formal reasoning of transformers.
We showed that $\empt$ is undecidable for classes of TE recently considered in research on the expressiveness of different 
transformer models  (Theorem~\ref{sec:undec;thm:undec} and Theorem~\ref{sec:undec;thm:logundec}). This implies 
that formal reasoning is impossible as soon as we consider classes of TE that are at least as expressive as the classes
considered in these results. 
Additionally, we identified two ways to enable formal reasoning for TE: by bounding the length of inputs 
(Theorem~\ref{sec:dec;thm:boundedword}) or by considering quantized TE, where computations and 
parameters are limited by fixed-width arithmetic (Theorem~\ref{sec:decidability;thm:nexptime}). These 
imply that formal reasoning is possible for TE classes that are at most as expressive as those in our 
results. Thereby, we assume that TE expressiveness is the primary factor influencing computability or 
complexity bounds, rather than specific safety or interpretability assumptions. However, in both cases, $\empt$ remains 
computationally difficult (Theorems~\ref{sec:dec;thm:boundedword} 
and~\ref{sec:decidability;thm:nexptimehard}). Again, these results apply only to TE classes at least as 
expressive as those we considered.
While our results provide an initial framework for understanding the possibilities and challenges of formal 
reasoning for transformers, there is room for more detailed investigations. Our undecidability and hardness 
results rely on normalizations realized by the hardmax function, and it’s unclear whether similar results hold 
when using the commonly employed softmax function. Additionally, further exploration of the interplay between 
the embedding function and the internal structure of the TE is of interest. We expect that less expressive 
embeddings require a richer attention mechanism, but it’s unclear where the limits lie regarding the 
undecidability of the satisfiability problem. Regarding our decidability and upper complexity bounds, examining 
the specifics of particular fixed-width arithmetics could be practically beneficial. 
While this wouldn’t change our 
overall results, it could provide tighter time-complexity estimates valuable for certain formal reasoning 
applications.

\bibliography{references}
\bibliographystyle{iclr2025_conference}

\appendix


\section{Tiling Problems}
\label{app:tiling}

We make use of particular tiling problems in order to prove lower bounds on the complexity
and decidability of $\empt[\transC]$ for different classes $\transC$.

A \emph{tiling system} is an $\mathcal{S} = (S,H,V,t_I,t_F)$ where $S$ is a finite set; its elements
are called \emph{tiles}. $H,V \subseteq S \times S$ define a horizontal, resp.\ vertical 
matching relation between tiles, and $t_I,t_F$ are two designated \emph{initial}, resp.\ 
\emph{final} tiles in $S$. 

Problems associated with tiling systems are typically of the following form: given a discrete convex plain
consisting of cells with horizontal and vertical neighbors, is it possible to cover the plane with tiles
from $S$ in a way that horizontally adjacent tiles respect the relation $H$ and vertically adjacent tiles
respect the relation $V$, together with some additional constraints about where to put the initial and
final tile $t_I, t_F$. Such tiling problems, in particular for rectangular planes, have proved to be 
extremely useful in computational complexity, cf.\ (\cite{MR0216954,Boas97}), since they can be seen as 
abstract versions of halting problems.  

We need a variant in which the plane to be tiled is of triangular shape. The \emph{$n$-th triangle} is 
$\mathcal{O}_n = \{ (i,j) \in \nats \times \nats \mid j \le i \le n \}$ for $n > 0$. An 
($\mathcal{S}$)-\emph{tiling} of $\mathcal{O}_n$ is a function $\tau : \mathcal{O}_n \to S$ s.t.
\begin{itemize}
\item $(\tau(i,j),\tau(i,j+1)) \in H$ for all $(i,j) \in \mathcal{O}$ with $j < i \le n$, 
\item $(\tau(i,j),\tau(i+1,j)) \in V$ for all $(i,j) \in \mathcal{O}$ with $j \le i < n$. 
\end{itemize}
Such a tiling a \emph{successful}, if additionally $\tau(0,0) = t_I$ and $\tau(i,i) = t_F$ 
for some $(i,i) \in \mathcal{O}_n$.

The \emph{unbounded octant tiling problem} (\unbOTP) is: given a tiling system $\mathcal{S}$, decide whether a
successful $\mathcal{S}$-tiling of $\mathcal{O}_n$ exists for some $n \in \nats$. 
The \emph{bounded octant tiling problem} (\bOTP) is: given a tiling system $\mathcal{S}$ and an 
$n \ge 1$, decide whether a successful $\mathcal{S}$-tiling of $\mathcal{O}_{n}$ 
exists. Note that here, $n$ is part of the input, and that it can be represented differently, for example
in binary or in unary encoding. We distinguish these two cases by referring to $\bOTP_{\mathsf{bin}}$ and
$\bOTP_{\mathsf{un}}$.

It is well-known that \unbOTP is undecidable (\cite{Boas97}). It is also not hard to imagine
that $\bOTP_{\mathsf{un}}$ is NP-complete while $\bOTP_{\mathsf{un}}$ is \NEXPTIME-complete. In 
fact, this is well-known for the variants in which the underlying plane is not a triangle of height
$n$ but a square of height $n$ (\cite{Boas97}). The exponential difference incurred by the more compact
binary representation of the input parameter $n$ is best seen when regarding the upper complexity
bound for these problems: given $n$, a nondeterministic algorithm can simply guess all the $n^2$ many
tiles of the underlying square and verify the horizontal and vertical matchings in time $\mathcal{O}(n^2)$. 
If $n$ is encoded unarily, i.e.\ the space needed to write it down is $s := n$, then the time needed for 
this is polynomial in the input size $s$; if $n$ is encoded binarily with space $s := \lceil \log n\rceil$ 
then the time needed for this is exponential in $s$. 

It then remains to argue that the tiling problems based on triangular planes are also NP- resp.\ 
\NEXPTIME-complete. Clearly, the upper bounds can be established with the same guess-and-check
procedure. For the lower bounds it suffices to observe that hardness of the tiling problems for the
squares is established by a reduction from the halting problem for Turing machines (TM) such that a square
of size $n \times n$ represents a run of the TM of length $n$ as a sequence of rows, and each row 
represents a configuration of the TM using at most $n$ tape cells. This makes use of the observation
that the space consumption of a TM can never exceed the time consumption. Likewise, assuming that
a TM always starts a computation with its head on the very left end of a tape, one can easily observe
that after $i$ time steps, it can change at most the $i$ leftmost tape cells. Hence, a run of a TM
can therefore also be represented as a triangle with its first configuration of length 1 in row 1,
the second of length 2 in row 2 etc.

At last, we consider two slight modifications of these two problems which are easily seen 
to preserve undecidability resp.\ NP- and \NEXPTIME-completeness. The \emph{unbounded octant tiling-word problem}
(\unbOTWP) is: given some $\mathcal{S} = (S,H,V,t_I,t_F)$, decide whether there is a word 
$t_{0,0},t_{1,0},t_{1,1},t_{2,0},t_{2,1},t_{2,2},\ldots,$ $t_{n,n} \in S^*$ for some $n \in \nats$, s.t.\
the tiling $\tau$ defined by $\tau(i,j) := t_{i,j}$ comprises a successful tiling of $\mathcal{O}_n$. 
The two variants of the \emph{bounded octant tiling-word problem} are both: given some 
$\mathcal{S}$ as above and $n$, decide whether such a word exists. Note that, again, here $n$ is
an input parameter, and so its representation may affect the complexity of the problem, leading 
to the distinction between $\bOTWP_{\mathsf{bin}}$ with binary encoding and $\bOTWP_{\mathsf{un}}$
with unary encoding. 

\begin{theorem}
\label{thm:tiling_problems}
\ \begin{itemize}
	\item[a)] \unbOTWP is undecidable ($\Sigma_0^1$-complete).
	\item[b)] $\bOTWP_{\mathsf{bin}}$ is \NEXPTIME-complete.
	\item[c)] $\bOTWP_{\mathsf{un}}$ is NP-complete.
	\end{itemize}
\end{theorem}
\begin{proof}
	(a) It should be clear that a tiling problem and its tiling-word variant (like \unbOTP and \unbOTWP) are 
	interreducible since they only differ in the formulation of how the witness for a successful tiling
	should be presented. So they are essentially the same problems. Undecidability of \unbOTP and, thus,
	\unbOTWP is known from (\cite{Boas97}), the $\Sigma^1_0$-upper bound can be obtained through a
	semi-decision procedure that searches through the infinite space of $\mathcal{O}_n$-tiling for any
	$n > 1$. This justifies the statement in part (a) of Thm.~\ref{thm:tiling_problems}.
	
	(b) With the same argument as in (a) t suffices to consider $\bOTP_{\mathsf{bin}}$ instead of 
	$\bOTWP_{\mathsf{bin}}$. The upper bound is easy to see: a nondeterministic procedure can easily guess a tiling
	for $\mathcal{O}_n$ and verify the horizontal and vertical matching conditions, as well as the
	use of the initial and final tile in appropriate places. This is possible in time $\mathcal{O}(n^2)$,
	resp.\ $\mathcal{O}(2^{2\log n})$ which is therefore exponential in the input size $\lceil \log n \rceil$
	for binarily encoded parameters $n$. This shows inclusion in \NEXPTIME. 
	
	For the lower bound we argue that the halting problem for nondeterministic, exponentially-time
	bounded TM can be reduced to $\bOTP_{\mathsf{bin}}$: given a nondeterministic TM $\mathcal{M}$ over 
	input alphabet $\Sigma$ and tape alphabet $\Gamma$ that halts after at most time $2^{p(n)}$ steps 
	on input words of length $n$ for some polynomial $n$, and a word $w \in \Sigma^*$, we first construct
	a TM $\mathcal{M}_w$ that is started in on the empty tape and begins by writing $w$ onto the tape
	and then simulates $\mathcal{M}$ on it. This is a standard construction in complexity theory, and
	it is easy to see that the running time of $\mathcal{M}_w$ is bounded by a function $2^{p'(|w|)}$
	for some polynomial $p'$. With the observation made above, a computation of $\mathcal{M}_w$ can 
	be seen as a sequence of configurations $C_1,\ldots,C_{p'(|w|)}$, with $|C_i| = i$. This does not
	directly define a tiling system, instead and again by a standard trick, cf.\ (\cite{Boas97}) or
	\cite[Chp.~11]{TLCS-DGL16}, one compresses three adjacent tape cells into one tile in order to
    naturally derive a horizontal matching relation from overlaps between such triples and a vertical
    matching relation from the TM's transition function. At last, let $n' := p'(|w|)$. It is then a 
    simple exercise to verify that a valid tiling of the triangle $\Delta_{n'}$ corresponds to an 
    accepting run of $\mathcal{M}$ on $w$ and vice-versa, which establishes \NEXPTIME-hardness. 
	
    (c) This is down exactly along the same lines as part (b), but instead making use of the fact that,
    when $n$ is given in unary encoding, $p(n)$ is polynomial in the size of the representation of 
    $n$, and hence, the time needed for the guess-and-check procedure in the upper bound is only
    polynomial, and for the lower bound we need to assume that the running time of the TM is polynomially
    bounded. Thus, we get NP-completeness instead of \NEXPTIME-completeness.   
\end{proof}

\section{Proofs of Section~\ref{sec:undec}}
\label{app:undec_proofs}

In the following, we give formal proof for the undecidability results of Section~\ref{sec:undec}. To do so, we make use of classical Feed-Forward Neural Networks.

\paragraph{Feed-Forward Neural Network}
A \emph{neuron} $v$ is a computational unit computing a function $\reals^m \rightarrow \reals$ by 
$v(x_1, \dotsc, x_m) = \sigma(b + \sum_{i=1}^m w_i x_i)$ where $\sigma$ is a function called \emph{activation}
and $b,w_i$ are parameters called \emph{bias} resp.\ \emph{weight}. A \emph{layer} $l$ is a tuple of nodes $(v_1, \dotsc, v_n)$ 
where we assume that all nodes have the same input dimensionality $m$. Therefore, $l$ computes a function 
$\reals^m \rightarrow \reals^n$. We call $n$ the \emph{size of layer $l$}.
Let $l_1$ be a layer with input dimensionality $m$ and $l_k$ a layer of size $n$. 
A \emph{Feed-Forward Neural Network (FNN)} $N$ is a tuple $(l_1, \dotsc, l_k)$ of layers where we assume that for all $i \leq k-1$
holds that the size of $l_i$ equals the input dimensionality of $l_{i+1}$. Therefore, $N$ computes a function $\reals^m \rightarrow \reals^n$
by processing an input layer by layer.

In particular, we use specific FNN with $\relu(x) = \max(0,x)$ activations, called \emph{gadgets}, 
to derive lower bounds in connection with the expressibility of transformers. We denote the class of all 
FNN with $\relu$ activations by $\pwlfnn$.
\begin{lemma}
	\label{app:undec_proofs;lem:gadgets}
	Let $k \in \reals^{>0}$. There are basic gadgets
	\begin{enumerate}
			\item $N_{|\cdot|} \in \pwlfnn$ computing $N_{|\cdot|}(x) = |x|$, 
			\item $N_< \in \pwlfnn$ computing a function $\reals^2 \rightarrow \reals$ such that
			$N_<(x_1,x_2) = 0$ if $(x_1 + 1) - x_2 \leq 0$, $N_<(x_1,x_2) = (x_1 + 1) - x_2 $ if $(x_1 + 1) - x_2  \in (0;1)$
			and $N_<(x_1,x_2) = 1$ otherwise,
			\item $N_= \in \pwlfnn$ computing a function $\reals^2 \rightarrow \reals$ such that
			$N_= (x_1, x_2) = 0$ if $x_1 - x_2 = 0$, $N_= (x_1, x_2) = |x_2-x_1|$ if $|x_2 - x_1| \in (0;1)$
			and $N_=(x_1,x_2) = 1$ otherwise,
			\item $N_\rightarrow \in \pwlfnn$ computing a function $\reals^2 \rightarrow \reals$
			such for all inputs $x_1, x_2$ with $x_1 \in \{0,1\}$ and $x_2 \in [0;k]$ holds 
			$N_\rightarrow(x_1, x_2) = 0$ if $x_1 = x_2=0$ or $x_1=1$ 
			and $N_\rightarrow(x_1, x_2) = \relu(x_2)$ otherwise.
		\end{enumerate}
\end{lemma}
\begin{proof}
	Let $N_{|\cdot|}$ be the minimal FNN computing $\relu(\relu(-x) + \relu(x))$, let
	$N_<$ be the minimal FNN computing $\relu(f_<(x_1, x_2) - f_<(x_1, x_2+1))$ where $f_<(y_1, y_2) = \relu(y_1 - y_2 + 1)$
	and let $N_=$ be the minimal FNN computing $\relu(f_=(x_1, x_2) - f_= (x_1 + 1, x_2) + f_=(x_2,x_1) -  f_= (x_2 + 1, x_1))$ 
	where $f_=(y_1, y_2) = \relu(y_2 - y_1)$. The claims of the lemma regarding these gadgets are straightforward given their functional form.
	Let $N_\rightarrow$ be the minimal FNN computing $\relu(\relu(x_2) - k\cdot\relu(x_1))$. As stated in the lemma, we assume that $x_1 \in \{0,1\}$
	and $x_2 \in [0;k]$. Then, $-k\cdot \relu(x_1)$ is $-k$ if $x_1=1$ and $0$ if $x_1=0$. Thus, $N_\rightarrow$ is guaranteed to be $0$ if 
	$x_1 = 1$ and otherwise it depends on $x_2$. This gives the claim regarding gadget $N_\rightarrow$.
\end{proof}

We will combine gadgets in different ways. Let $N_1$ and $N_2$ be FNN with the same input dimensionality $m$ and output
dimensionality $n_1$ respectively $n_2$. 
We extend the computation of $N_1$ to functions $\reals^{m'} \rightarrow \reals^{n_1}$ 
with $m < m'$ by weighting additional dimensions with $0$ in the input layer. Given a set of input dimensions
$x_1, \dotsc, x_{m'}$, we denote the effective dimensions $x_{i_1}, \dotsc, x_{i_{m}}$ with pairwise
different $i_j \in \{1, \dotsc, m'\}$ by $N^{x_{i_1}, \dotsc, x_{i_{m}}}_1$. Formally, this means that
$N^{x_{i_1}, \dotsc, x_{i_{m}}}_1(x_1, \dotsc, x_{m'}) = N_1(x_{i_1}, \dotsc, x_{i_m})$ for all inputs.
We denote the FNN consisting of $N_1$ and $N_2$ placed next to each other 
by $N_1 \para N_2$. Formally, this is done by combining $N_1$ and $N_2$ layer by layer using $0$ weights in intersecting connections. 
Then, $N_1 \para N_2$ computes $\reals^m \rightarrow \reals^{n_1 + n_2}$ given by $N_1 \para N_2(\bs{x}) = (N_1(\bs{x}), N_2(\bs{x}))$.
We generalize this operation to $k$ FNN $N_1 \para \dotsb \para N_k$ in the obvious sense. Let $N_3$ be an FNN with input
dimensionality $n_1$ and output dimensionality $n_3$. We denote the FNN consisting of $N_1$ and $N_3$ placed sequentially
by $N_3 \circ N_1$. Formally, this is done by connecting the output layer of $N_1$ with the input layer of $N_3$. Then, $N_3 \circ N_1$
computes $\reals^m \rightarrow \reals^{n_3}$ given by $N_3 \circ N_1 (\bs{x}) = N_3(N_1(\bs{x}))$. 

We also consider specific gadgets needed in the context of tiling problems.
\begin{lemma}
	\label{app:undec_proofs;lem:tilegadgets}
	Let $S \subseteq \nats$ be a finite set and $R \subseteq S^2$. There is 
	FNN $N_R \in \pwlfnn$ computing $\reals^2 \rightarrow \reals$ such that $N_R(x_1,x_2) \in \{0,1\}$
	if $(x_1,x_2) \in S^2$ and $N_R(x_1, x_2) = 0$ iff $(x_1, x_2) \in R$ and there is 
	$N_{=t} \in \pwlfnn$ for each $t \in S$ computing $\reals \rightarrow \reals$ such that
	$N_{=t}(x) \in \{0,1\}$ for each $x \in \nats$ and $N_{=t}(x) = 0$ iff $x = t$.
\end{lemma}
\begin{proof}
	Let $S \subseteq \nats$ be finite, $R \subseteq S^2$ and $t \in S$. First, consider
	$N_{=t}$. Let $N_t$ be the minimal FNN computing $\relu(0\cdot x + t)$ and $N_{\mathit{id}}$
	be the minimal FNN computing $(\relu(x), -\relu(-x))$. 
	Obviously, $N_t$ computes the constant $t$ function and
	$N_{\mathit{id}}$ computes the identity in the form of two dimensional vectors. Let $N_{=t}$ be given by the minimal FNN computing
	$N_= \circ (N_{\mathit{id}} \para N_t)$ with the slight alteration that the two output dimensions of $N_{\mathit{id}}$ 
	are connected to the first dimension of $N_=$. Then, the claim of the lemma regarding $N_{=t}$ follows from 
	Lemma~\ref{app:undec_proofs;lem:gadgets} and the operations on FNN described in Appendix~\ref{app:undec_proofs}.
	
	Now, consider $N_R$. Given some $s \in S$ let $R[s] = \{r \mid (s,r) \in R\}$. 
	Let $N^k_\land$ be the minimal FNN computing $\relu(x_1 + \dotsb + x_k)$. 
	Furthermore, let $N_{\in T}$ for some set $T \subseteq S$ be the minimal FNN such that 
	$N_{\in T} (x) = 0$ if $x \in T$ and $N_{\in T} (x) = 1$ if $x \in S\setminus T$. A construction
	for $N_{\in T}$ is given in Theorem 4 in (\cite{SalzerL23}). According to this construction,
	$N_{\in T}$ consists of three layers and is polynomial in $T$. In the case that $T = \emptyset$ we assume that 
	$N_{\in \emptyset}$ is the constant $1$ function represented by a suitable FNN.
	Then, $N_R$ is given by $N^{|S|}_\land \circ ((N_\rightarrow \circ (N_{=s_1} \para N_{\in R[s_1]})) \para \dotsb \para (N_\rightarrow \circ (N_{=s_{|S|}} \para N_{\in R[s_{|S|}]})))$
	for some arbitrary order on $S$ with the slight alteration that $N_R$ has two input dimensions, meaning that each subnet $(N_{=s_i} \para N_{\in R[s_i]})$
	is connected to the same two input dimensions. Again, the claim of the lemma regarding $N_R$ follows from 
	Lemma~\ref{app:undec_proofs;lem:gadgets} and the operations on FNN described in Appendix~\ref{app:undec_proofs}.
\end{proof}

Given these understandings of gadgets, we are set to formally prove the results of Section~\ref{sec:undec}.
\begin{proof}[Proof of Lemma~\ref{sec:undec;lem:decode}]
	Let $w = t_{0,0} t_{1,0} t_{1,1} t_{2,0} \dotsb t_{m,n} \in S^+$ as stated in the lemma and assume some order $a_i$ on $S$.  
	Furthermore, let $\emb(a_i, 1) = (1, 1, 1, 1, i)$ and $\emb(a_i, j) = (0, 1, j, \sum_{h=0}^{j} h, i)$ if $j > 1$. Let $\emb(w) = \bs{x}^0_1 \dotsb \bs{x}^0_k$. 
	In the following, we build two layers $l_1$ and $l_2$ using components allowed in $\transCundec$, satisfying 
	the statement of the lemma.
	Layer $l_1$ consists of a single attention head $\att_{1,1} = (\score_{1,1}, \pool_{1,1})$. The scoring function  is given
	by $\score_{1,1}(\bs{x}^0_i,\bs{x}^0_j) =  N_{1,1}(\langle Q_{1,1}\bs{x}^0_i, K_{1,1}\bs{x}^0_j\rangle)$ where $Q_{1,1} = [(0,0, -1, 0, 0), (0,1, 0, 0,  0), (0,1, 0, 0,  0)]$
	and $K_{1,1} = [(0,1, 0, 0, 0), (0,1, 0, 0, 0), (0,0, 0, 1,  0)]$ and $N (x) = -\relu(x)$. 
	We have that $\score_{1,1}(\bs{x}^0_i,\bs{x}^0_j) = -\relu((\sum_{h=0}^{j} h) - (i-1))$ and it follows that $\score_{1,1}(\bs{x}^0_i,\bs{x}^0_j) = 0$ if $\sum_{h=0}^{j} h \leq i-1$ and otherwise we have that $\score_{1,1}(\bs{x}^0_i,\bs{x}^0_j) < 0$. The pooling function is specified by the matrix $W_{1,1} = [(1,0,0,0,0)]$ and uses $\hmax$ as normalisation function. 
	The combination $\comb_1$ function is given by the FNN 
	$N_1(x_1, \dotsc, x_5, y) = \relu(x_2) \para \dotsb \para \relu(x_5) \para \relu(y)$.
	Given a position $\bs{x}^0_i$, the attention head $\att_{1,1}$ attends to 
	all positions $\bs{x}^0_j$ satisfying 
	$\sum_{h=0}^{j} h \leq i-1$. This is due to the way $\score_{1,1}$ is build. Then, $\att_{1,1}$ computes $\frac{1}{l}$ using 
	$\pool_{1,1}$ where $l$ is the number of positions $\att_{1,1}$ attends to. Here, we exploit the fact that only the first position
	$\bs{x}^0_1$ has a non-zero entry in the its first dimension and that for all $i$ head $\att_{1,1}$ attends to $\bs{x}^0_1$.
	Finally, $\comb_1$ simply stacks the old vector $\bs{x}^0_i$ onto the value $\frac{1}{l}$, but leaves out the first dimension of $\bs{x}^0_i$.
	Let $l_1(\emb(w)) = \bs{x}^1_1 \dotsb \bs{x}^1_k$.
	Layer $l_2$ consists of a single attention head $\att_{2,1} = (\score_{2,1}, \pool_{2,1})$. The scoring function 
	$\score_{2,1}$ is given by $N_{2,1}(\langle Q_{2,1}\bs{x}^1_i, K_{2,1}\bs{x}^1_j\rangle)$ where $Q_{2,1} = [(0,0,0,0,1)]$, 
	$K_{2,1} = [(0,1,0,0,0)]$ and $N_{2,1}(x) = -\relu(\relu(x-1) + \relu(1-x))$. We have that $\score_{2,1}(\bs{x}^1_i, \bs{x}^1_j) = 0$ if
	$\frac{1}{l} \cdot j = 1$ where $\frac{1}{l}$ is the fifth dimension of $\bs{x}^1_i$ and otherwise $\score_{2,1}(\bs{x}^1_i, \bs{x}^1_j) < 0$.
	The pooling function $\pool_{2,1}$ is specified by 
	$W_{2,1} = [(0,1,0,0,0), (0,0,1,0,0)]$ and 
	uses $\hmax$ as normalisation. The combination $\comb_2$  is given by
	the FNN $N_2(x_1, \dotsc, x_5, y_1,y_2) = \relu(x_1) \para \relu(x_2) \para \relu(y_1) \para \relu(x_2 - y_2 - 1) \para \relu(x_4)$. 
	Given a position $\bs{x}^1_i$, the attention head $\att_{2,1}$ attends to the position $j$, where $\frac{1}{l} \cdot j = 1$. Relying on
	our arguments regarding the computation of $l_1$, this is the position $j$ satisfying $\max_j(\sum_{h=0}^{j} h \leq i-1)$. However, this
	$j$ is equal to the row index $r(i)$ of the decomposition of $i$ based on the inversion of Cantor's pairing function. Thus, we have that $r(i) = j$.
	Furthermore, we have that $c(i) = (i-1) - (\sum_{h=0}^{j} h)$, which is computed by $\relu(x_2 - y_2 - 1)$ in the combination function $\comb_2$.
	Overall, we see that $l_2(l_1(\emb(w)))$ gives the desired result.
\end{proof}

\begin{proof}[Proof of Lemma~\ref{sec:undec;lem:lin}]
	Let $f$ be as stated in the lemma. By definition of $\transCundec$, the scoring function of $\att_f$ is of the form $N(\langle Q\bs{x}_i, K\bs{x}_j\rangle)$ and the 
	normalisation is $\hmax$. Let $Q = [(a_1, \dotsc,  a_k), (b, 0, \dotsc, 0), (1, 0, \dotsc, 0)]$, $K = [(1, 0, \dotsc, 0), (1, 0, \dotsc, 0), (0,-1,0, \dotsc, 0)]$ and 
	$N$ be the minimal FNN computing $N(x) = -\relu(N_{|\cdot|}(x)) = -|x|$ where $N_{|\cdot|}$ is given by Lemma~\ref{app:undec_proofs;lem:gadgets}. 
	Overall, this ensures that the scoring is given by $\score(\bs{x}_i, \bs{x}_j) =-|f(\bs{x}_i) - j|$. Then, the statement of the lemma follows from the
	fact that $\hmax$ attends to the maximum, which is $0$ given this scoring, and that $j \in \nats$ is unique for each $\bs{x}_j$. 
\end{proof}

\begin{lemma}
	\label{app:undec_proofs;lem:leq}
	There is attention head $\att_\leq$ in $\transCundec$ 
	such that for all sequences $\bs{x}_1, \dotsc, \bs{x}_m$ where all $\bs{x}_i = (1, i, \bs{y}_i)$ 
	the head $\att_\leq$ attends to $\{\bs{x}_1, \dotsc, \bs{x}_{i}\}$ given $i$. 
\end{lemma}
\begin{proof}
	By definition of $\transCundec$, the scoring function of $\att_f$ is of the form $N(\langle Q\bs{x}_i, K\bs{x}_j\rangle)$ and the 
	normalisation is $\hmax$. Let $Q = [(0, 1, 0, \dotsc,  0), (1, 0, \dotsc,  0)]$ and let $K$ be equal to $[(1, 0 \dotsc,  0), (0, -1, 0, \dotsc,  0)]$. 
	Furthermore, let 
	$N(x) = -\relu(x)$. We observe that $N$ outputs $0$ if $j \leq i$ and otherwise $N(x) < 0$.
	In combination with $\hmax$, this ensures that $\att_\leq$ behaves as stated by the lemma. 
\end{proof}
\begin{proof}[Proof of Theorem~\ref{sec:undec;thm:undec}]
	{
		\newcommand*{\T}{T_\mathcal{S}}
		\newcommand*{\err}{\mathit{err}}
		
		We prove the statement via reduction from $\unbOTWP$. Let $\mathcal{S} = (S,H,V,t_I,t_F)$
		be an instance of $\unbOTWP$ with $|S| = k$. W.l.o.g we assume that $S \subseteq \nats$. 
		Let $\T \in \transCundec$ be built the following way.
		$\T$ uses the embedding $\emb$ of transformer in $\transCundec$ specified in the beginning of Section~\ref{sec:undec}. 
		Furthermore, it has four layers. Layers $l_1$, $l_2$ are as in Lemma~\ref{sec:undec;lem:decode}. Layer 
		$l_3$ is given by $l_3=(\att_{\mathit{prev}}, \att_{\mathit{next}}, \att_{\mathit{step}}, \comb_3)$ 
		where  $\att_{\text{prev}}$, $\att_{\text{next}}$ and $\att_{\text{step}}$ are of 
		Lemma~\ref{sec:undec;lem:lin} whereby  
		$\mathit{prev}(x_1, \dotsc, x_5) = x_2 - 1$,
		$\mathit{next}(x_1, \dotsc, x_5) = x_2 + 1$
		and $\mathit{step}(x_1, \dotsc, x_5) = x_2 + x_3 + 1$. 
		We assume that all three attention heads use the identity matrix as linear maps in their respective pooling function. 
		$\comb_3$ is given by an FNN $N_3$ computing $\reals^{4 \cdot 5} \rightarrow \reals$. Let the input
		dimensions of $N_3$ be $x_{1,1}, \dotsc, x_{1,5}, x_{2,1}, \dotsc, x_{4,5}$. Then, $N_3$ is equal to 
		\begin{displaymath}
			\relu(x_{1,1}) \para \relu(x_{1,2}) \para N_a \para N_{b_1} \para N_{b_2} \para N_c \para N_d
		\end{displaymath}
		where $N_a = N_\rightarrow \circ (N^{x_{1,2}, x_{3,2}}_= \para N_=^{x_{1,3}, x_{1,4}})$,
		$N_{b_1} = N_\rightarrow \circ (N^{x_{1,2}, x_{2,2}}_= \para N_{=t_I}^{x_{1,5}})$,
		$N_{b_2} = N_\rightarrow \circ (N^{x_{1,2}, x_{3,2}}_= \para N_{=t_F}^{x_{1,5}})$, 
		$N_c = N_\rightarrow \circ (N^{x_{1,4}, x_{1,3}}_< \para N^{x_{1,5},x_{3,5}}_H)$
		and $N_d = N_\rightarrow \circ (N_<^{x_{1,3}, {x_{4,3}}}  \para N^{x_{1,5},x_{4,5}}_V)$
		using the gadgets and constructions described in Appendix~\ref{app:undec_proofs}. 
		Layer $l_4$ is given by  $l_4= (\att_{\text{leq}}, \comb_4)$
		where $\att_{\text{leq}}$ attends to $\{\bs{x}_1, \dotsc, \bs{x}_i\}$ given $i$
		and $\comb_4$ is given by the minimal FNN $N_4$ computing $\relu(x_3 + \dotsb + x_7)$.
		A formal proof for the existence of $\att_{\text{leq}}$ in $\transCundec$ is given in Lemma~\ref{app:undec_proofs;lem:leq}. 
		Furthermore, the output function $\out$ of $\T$ is given by the minimal FNN $N_\out$ computing $N(x_1) =
		\relu(1 - x_1)$.
		
		Let $w= t_1 \dotsb t_l \in S^*$ be some word over alphabet $S$. As defined above, we have that 
		$\emb(t_i,i) = (1,i, \sum_{j=0}^{i}j, k_i)$ where $k_i \in \{1, \dotsc, |S|\}$. 
		Consider $\bs{x}^2_1 \dotsb \bs{x}^2_m$, namely the sequence of vectors after propagating $w$
		through the embedding $\emb$ and layers $l_1,l_2$ of $\T$. As stated by Lemma~ \ref{sec:undec;lem:decode}, we have that 
		$\bs{x}^2_i = (1,i, r(i), c(i), k_i)$ where $r(i)$ and $c(i)$ are the row respectively column of tile $t_i$ 
		if we interpret $w$ as an encoded tiling. Note that all vectors $\bs{x}^3_i$ are non-negative due to the way $N_3$ is built. 
		In the following, we argue that all $\bs{x}^3_i = \bs{0}$ if and only if $w$ is a valid
		encoded tiling. Given this equivalence, the statement of the lemma follows immediately as $l_4$ simply sums up
		all vectors and dimensions (except for the first and second) of $\bs{x}^3_1, \dotsc, \bs{x}^3_m$  in $\bs{x}^4_m$ 
		and the output of $N_4$ indicates whether there was some non-zero value. 
		We fix some arbitrary $\bs{x}^2_i = (1,i, r(i), c(i), k_i)$. 
		Then, $\bs{x}^3_i= N_3(\bs{x}^2_i, \bs{x}^2_{i_\mathit{prev}}, \bs{x}^2_{i_\mathit{next}}, \bs{x}^2_{i_\mathit{step}})$ where 
		$i_\mathit{next} = i+1$ if $i < m$ and $m$ otherwise, $i_\mathit{prev} = i-1$ if $i > 1$ and $1$ otherwise and 
		$i_\mathit{step} = i + r(i) + 1$ if $ i < m - r(i) - 1$ and $m$ otherwise. 
		
		Consider property $(a)$ and subnetwork $N_a$. With the understanding gained in Appendix~\ref{app:undec_proofs},  
		$N^{x_{1,2}, x_{3,2}}_=$ outputs $0$ iff $x_{1,2} = x_{3,2}$. These dimensions correspond to positions 
		$i$ and $i_\mathit{next}$, which are only equal if $i = m$ (Lemma~\ref{sec:undec;lem:lin}). Furthermore,
		the property of $N_\rightarrow$ stated by Lemma~\ref{app:undec_proofs;lem:gadgets} is given as the output of $N_=$ is guaranteed to be in 
		$[0;1]$ and the values of $x_{1,2}$ and $x_{3,2}$ are guaranteed to be in $\nats$. In summary, this ensures that
		the third dimension of $\bs{x}^3_m$ is $0$ iff $r(m) = c(m)$. For other positions the third dimension is always $0$ since 
		$N_\rightarrow$ outputs $0$ in these cases due to the fact that $N^{x_{1,2}, x_{3,2}}_=$ equals $1$. 
		Analogously, $N_{b_1}$ and $N_{b_2}$ ensure that $t_1 = t_I$ and $t_m = t_F$ and, thus, property (b) iff the fourth
		and fifth dimensions in all positions are equal to $0$.
		Consider properties (c) and (d) described above and assume that property (a) holds. 
		These two properties are non-local in the sense that they depend on at least two positions in
		$\bs{x}^2_1 \dotsb \bs{x}^2_m$. Consider the subnet $N_c$. By construction and the gadgets described 
		in Appendix~\ref{app:undec_proofs}, we have that $N_c$ outputs $0$ if $c(i) < r(i)$ and $(t_i, t_{i+1}) \in H$
		or if $c(i) = r(i)$, which means that tile $t_i$ is rightmost in its corresponding row. Otherwise the value computed by $N_c$ is greater
		than $0$. Analogously, subnet $N_d$ checks whether vertically stacked tiles do match. In summary, this ensures that
		the sixth and seventh dimension of each $\bs{x}^3_i$ is equal to $0$ if and only if properties (c) and (d) hold.
	}
\end{proof}

\begin{proof}[Proof of Theorem~\ref{sec:undec;thm:logundec}]
	In the same manner as in the proof of Theorem~\ref{sec:undec;thm:undec}, 
	we prove the statement via reduction from $\unbOTWP$. The reduction is exactly the 
	same, namely given an $\unbOTWP$ instance $\mathcal{S} = (S, H, V, t_I, t_F)$ we build TE 
	$T_{\mathcal{S}}$ which recognizes exactly those words $w$ representing a valid encoded 
	tiling of $\mathcal{S}$. For details, see the proof of Theorem~\ref{sec:undec;thm:undec}.
	
	Given the correctness arguments for $T_{\mathcal{S}}$ in Theorem~\ref{sec:undec;thm:undec}, it 
	is left to argue that $T_{\mathcal{S}}$ works as intended, 
	despite the fact that it works over some FA $F$ using at most $\mathcal{O}(\log(\max(|S|,n)))$ bits
	where $n$ is the length of an input word. We choose $F$ such that overflow 
	situations do not occur in any computation $T_{\mathcal{S}}(w)$ and rounding 
	is handled such that $T_{\mathcal{S}}$ works as intended.
	Throughout this proof, we use $\log(n)$ 
	Namely, given a word $w$ with $|w| = n$ assume that $F$ uses $m = \rdown{4\log(\max(|S|, n))}+2$ bits and rounds values off to the nearest representable number.
	We denote the value resulting from rounding $x$ off in arithmetic $F$ by $\rdown{x}_F$.
	We assume that there is an extra bit that is used as a sign bit and that at least $\rdown{3\log(n)} + 1$ bits can be 
	used to represent integer and at least $\rdown{\log(n)} + 1$ bits can be used to represent fractional 
	parts. Note that this is a reasonable assumption for all common FA, like fixed-point or 
	floating-point arithmetic. Furthermore, it is clearly the case that $m \in 
	\mathcal{O}(\log(\max(|S|,n)))$. To ease our arguments and notation from here on, we assume w.l.o.g.\  that we represent $n$ using $\log(n)$ instead of $\rdown{log(n)} +1$.
	
	Per definition, $T_\mathcal{S}$ uses the embedding function 
	$\emb(a_k,0) = (1,1, 0, 0,  k)$ and $\emb(a_k,i) = (0,1, i, \sum_{j=0}^{i} j , k)$. First, we assume that each $k$,
	namely the value representing a specific tile from $S$, is a unique, positive value. This is possible as $F$ 
	uses $m > \log(|S|)$ bits. Furthermore, we see that $\emb$, especially the sum $ \sum_{j=0}^{i} j = \frac{i(i+1)}{2} 
	\leq i^2$, works as intended up to $i = n$ due to the fact that $F$ uses more than $m > 2\log(n)$ bits to represent integer parts. 
	Next, consider layer $l_1$ and $l_2$ of Lemma~\ref{sec:undec;lem:decode}. Layer $l_1$
	consists of a single attention head $\att_{1,1}$. Here, the only crucial parts are the computation of value $\frac{1}{l}$ in $\pool_{1,1}$ for a position $i$. Per definition,
	$l$ corresponds to the number of positions $j$ such that $\sum_{h=0}^{j} h \leq i-1$.
	As $i$ is bounded by $n$, this inequality can only be satisfied by positions $j$ for which
	$j \leq \sqrt{n}$ holds. As $T_{\mathcal{S}}$ uses $\hmax$ to count the positions for which
	this inequality holds, $l$ is bounded by $\sqrt{n}$. Next, we 
	observe that $\rdown{\frac{1}{l}}_F = \frac{\rdown{2^{\log(n)}\frac{1}{l}}}{2^{\log(n)}}= \frac{\rdown{\frac{n}{l}}}{n}$, namely 
	the general understanding of rounding off where
	we use $\log(n)$ bits to represent fractions. However, this
	gives that for all $1 \leq l_1 < l_2 \leq \sqrt{n}$ that
	$\rdown{\frac{1}{l_1}}_F \neq \rdown{\frac{1}{l_2}}_F$ as
	$\rdown{\frac{n}{l_1}} \neq \rdown{\frac{n}{l_2}}$
	holds for all $l_1 < l_2 \leq \sqrt{n}$. This means, that
	it is ensured by $F$ that $\frac{1}{l}$ is uniquely representable.
	
	Next, the only crucial part in $l_2$ is the computation of the product $\frac{1}{l} \cdot j$, which is used to determine
	the position $j$ for which $\frac{1}{l} \cdot j = 1$ in $\score_{2,1}$, which is obviously given by position $l$. 
	This equality is no longer guaranteed to
	exist if we consider $\rdown{\frac{1}{l}}_F \cdot j$. However,
	due to the monotonicity of $\rdown{\frac{1}{l}}_F$ for 
	$l \leq \sqrt{n}$ and that the maximum round of error is given
	by $\frac{1}{2^{\log(n)}}$, we have that the $j=l$ produces the
	value closest to $1$ in the product $\frac{1}{l} \cdot j$.
	Taking a look at $\score_{2,1}$, this ensures that 
	$l$ is still the position that $\att_{2,1}$ attends to.
	Therefore, the statement of Lemma~\ref{sec:undec;lem:decode} is 
	still valid for $T_{\mathcal{S}}$ working over $F$. 
	We observe that all values of some
	vector $\bs{x}^2_j$ after layer $l_2$ are positive integers whose magnitude is bounded 
	by $n^2$.
	
	Now, consider layer $l_3$ and $l_4$. From the proof of 
	Theorem~\ref{sec:undec;thm:undec} we see that the gadgets at most sum up
	two values or compute a fraction of the form $\frac{i+j}{2}$ and $\frac{i-j}{2}$ 
	(in gadgets $N_H$ or $N_V$). Both can safely be done with at least $3\log(n)$ bits for 
	integer and $\log(n)$ for fractional parts, as all previously computed values, up to layer 
	$l_2$, in a computation of $T_{\mathcal{S}}(w)$ are representable using $2\log(n)$ bits.
	We observe that the values of the third to seventh dimension of some $\bs{x}^3_j$ are either $0$ or $1$. This is due to the fact that all values after layer $l_2$
	are guaranteed to be integers.
	Next, consider layer $l_4$. The computation done by $\att_\leq$ is safe (see Lemma~\ref{app:undec_proofs;lem:leq}) and the crucial step here is the computation of $\comb_4$ given by $\relu(x_3 + \dotsb + x_7)$. The values $x_i$ are all of the form 
	$\frac{i}{j}$ where $i$ is guaranteed to be $0$ or $1$ and $j$ is the normalisation
	induced by $\att_\leq$ from perspective of position $j$. However, this means $j$ is bounded by $n$ and, thus,
	$\rdown{\frac{i}{j}}_F > 0$ if and only if $i=1$ for all $j$
	due to the fact that $F$ allows for $\log(n)$ bits
	to represent fractional parts.
	Finally, $\out$ is trivially computable in $F$, which finishes the proof.
\end{proof}

\section{Proofs of Section~\ref{sec:dec}}
\label{app:dec_proofs}

\begin{proof}[Proof of Theorem~\ref{sec:dec;thm:boundedword}]
	The decidability and membership results of statements (\ref{sec:dec;thm:boundedword;prop:NP}) and (\ref{sec:dec;thm:boundedword;prop:NEXPTIME} )are sufficiently 
	argued in the proof sketch given in Section~\ref{sec:dec}.
	
	To prove the hardness results of statements (\ref{sec:dec;thm:boundedword;prop:NP}) and (\ref{sec:dec;thm:boundedword;prop:NEXPTIME}),
	we establish a reduction from $\bOTWP_{\mathsf{un}}$ respectively 
	$\bOTWP_{\mathsf{bin}}$: given some bounded word-tiling
	instance $(\mathcal{S}, n)$ we build an instance $(T_{\mathcal{S}},n)$ of $\bempt_{\mathsf{un}}$ respectively $\bempt_{\mathsf{bin}}$ where
	$T_\mathcal{S}$ is build as described in Theorem~\ref{sec:undec;thm:undec}. The only missing argument is that these reductions are polynomial. 
	In particular, this means that $T_{\mathcal{S}}$ must be built in polynomial time regarding the size of $(\mathcal{S}, n)$. 
	Therefore, we recall the proof of Theorem~\ref{sec:undec;thm:undec}.
	
	First, we see that the embedding function $\emb$ and the amount of layers of $T_{\mathcal{S}}$ is independent of $\mathcal{S}$ and $n$. 
	The first two layers $l_1$ and $l_2$ of $T_{\mathcal{S}}$ are specified in Lemma~\ref{sec:undec;lem:decode}. 
	Recalling the proof of Lemma~\ref{sec:undec;lem:decode}, we see that $l_1$ and $l_2$ each consist of a single attention head, 
	whose internal parameters like scoring, pooling or combination are independent of $(\mathcal{S}, n)$ as well.  Next, consider layer $l_3$.
	This layer consists of three attention heads $\att_{\mathit{prev}}$, $\att_{\mathit{next}}$ and $\att_{\mathit{step}}$ each given by the 
	template described in Lemma~\ref{sec:undec;lem:lin}, which again is independent of $(\mathcal{S},n)$. Additionally, $l_3$ contains
	the combination function $\comb_3$. This combination function is represented by a FNN $N_3$, using smaller FNN 
	$N_a$, $N_{b_1}$, $N_{b_2}$, $N_c$ and $N_d$ as building blocks. These are dependent on $\mathcal{S}$, as they are built
	using gadgets $N_{=t_I}$, $N_{=t_F}$, $N_H$ and $N_V$ where $t_I$, $t_F$, $H$ and $V$ are components of 
	$\mathcal{S}$. However, in the proof of Lemma~\ref{app:undec_proofs;lem:tilegadgets} we see that these gadgets 
	are at most polynomial in their respective parameter.  Layer $l_4$ and the output function, specified by FNN $N_{\mathit{out}}$,
	are again independent of $(\mathcal{S},n)$. In summary, the TE $T_{\mathcal{S}}$ is polynomial in $(\mathcal{S},n)$,
	which makes the reductions from $\expbOTWP$ und $\polybOTWP$ polynomial.
\end{proof}

Next, we address the proof of Lemma~\ref{sec:decidability;lem:shortword_softmax}. We need some preliminary, rather technical result first. 
Let $T$ be an TE and $w \in \Sigma^+$ be a word and consider the computation $T(w)$. 
Let $X_{T(w)}^0 =\emb(w)$ and $X_{T(w)}^i$ be the sequence of vectors occurring after the computation of layer $l_i$ of $T$.
Let $\bs{x}$ and $\bs{x}'$ be two vectors matching the dimensionality of $\score_{i,j}$ of $T$.
Overloading some notation, let $N_{w}(\bs{x},\bs{x}',i,j) = \norm_{i,j}(\score_{i,j}(\bs{x}, \bs{x}'), \score_{i,j}(\bs{x}, X_{T(w)}^{i-1}))$ where $\score_{i,j}(\bs{x}, X_{T(w)}^{i-1})$ is the vector of all scorings of $\bs{x}$ with sequence 
$X_{T(w)}^{i-1}$. We remark that it is not necessary that $\bs{x}$ or $\bs{x}'$ must occur in $X_{T(w)}^{i-1}$ for this to be well defined.
Again overloading some notation, let $P_{w}(\bs{x},i,j) = \pool_{i, j}(X_{T(w)}^{i-1}, \score_{i,j}(\bs{x}, X_{T(w)}^{i-1}))$.
\begin{lemma}
	\label{app:dec;lem:cutout}
	Let $T$ be a additive-periodical TE of depth $L$, maximum width $H$ and periodicity $p$ with $\norm_{i,j} \in \{\softmax, \hmax\}$ for all
	$i \leq L, j \leq H$, let $w = u_1u_{j_1} \dotsb u_{j_h}u_2 \in \Sigma^+$ where 
	$u_1, u_2 \in \Sigma^+$, all $u_{j_i} \in \Sigma^p$ and all $u_{j_i}$ also occur in $u_1$ or $u_2$ and
	let $\mathcal{X}$ be the set of all vectors occurring in any of the sequences $X^i_{T(w)}$. 
	If there are  indexes 
	$h_1 < h_2 \leq h$ such that for all $\bs{x},\bs{x}' \in \mathcal{X}, i \leq L, j \leq H$ holds that 
	$N_{u_1u_{j_1} \dotsb u_{j_{h_1}}}(\bs{x},\bs{x}',i,j) = N_{u_1u_{j_1} \dotsb u_{j_{h_2}}}(\bs{x},\bs{x}',i,j)$ and $P_{u_1u_{j_1} \dotsb u_{j_{h_1}}}(\bs{x},i,j) = P_{u_1u_{j_1} \dotsb u_{j_{h_2}}}(\bs{x},i,j)$ 
	then it holds that $N_{u_1u_{j_1} \dotsb u_{j_{h_1}} u_{j_{h_2+1}} \dotsb u_2}(\bs{x},\bs{x}',i,j) = N_{u_1 \dotsb u_2}(\bs{x},\bs{x}',i,j)$ and $P_{u_1u_{j_1} \dotsb u_{j_{h_1}} u_{j_{h_2+1}} \dotsb u_2}(\bs{x},i,j) = P_{u_1 \dotsb u_2}(\bs{x},i,j)$.
\end{lemma}
\begin{proof}
	Let $T$, $w$, $\mathcal{X}$, $h_1$ and $h_2$ be as stated above.
	We prove the statement via induction on the layers $l_i$. First,
	consider layer $l_1$ and fix some tuple $(\bs{x}, \bs{x}', 1, j)$.
	We first show that $N_{u_1u_{j_1} \dotsb u_{j_{h_1}} u_{j_{h_2+1}} \dotsb u_2}(\bs{x},\bs{x}',1,j) = N_{u_1 \dotsb u_2}(\bs{x},\bs{x}',1,j)$. 
	Assume that $\norm_{1,j}$ is given by $\softmax$. Then, $\norm_{1,j}$
	computes $\frac{e^{\score_{1,j}(\bs{x}, \bs{x}')}}{\sum_{\score_{1,j}(\bs{x}, X_{T(w')}^{0})} e^{s_{i'}}}$ for all words $w'$. 
	Obviously, the numerator in $N_{u_1 \dotsb u_{h_1}u_{h_2+1} \dotsb u_2}(\bs{x},\bs{x}',1,j)$ and $N_{u_1 \dotsb u_2}(\bs{x}, \bs{x}', 1,j)$ is equal. 
	By definition, we have that $\score_{i,j}$ is local in the sense that it
	compares vectors pairwise, producing the different scoring values 
	$s_{i'}$ independent of the overall word. Furthermore, due to the 
	fact that $\emb$ is additive-periodical, we have $X^0_{T(u_1u_{j_1} \dotsb u_{j_{h_1}} u_{j_{h_2+1}} \dotsb u_2)}$ and $X^0_{T(u_1 \dotsb u_2)}$ are equal in the sense that the vectors corresponding to 
	$u_{j_{h_2+1}} \dotsb u_2$ are equal. We refer to this property (*) later on. Using these observations 
	and that $N_{u_1u_{j_1} \dotsb u_{j_{h_1}}}(\bs{x},\bs{x}',1,j) = 
	N_{u_1u_{j_1} \dotsb u_{j_{h_2}}}(\bs{x},\bs{x}',1,j)$, we have that 
	the denominator is equal as well. Now, assume that $\norm_{1,j}$ 
	is given by $\hmax$. Then, 
	$\norm_{1,j}$ computes $\frac{f(\score_{1,j}(\bs{x}, \bs{x}'), \score_{1,j}(\bs{x}, X^0_{T(w')}))}{\sum_{\score_{1,j}(\bs{x}, X^{0}_{T(w')})} f(s_{i'}, \score_{1,j}(\bs{x}, X^{0}_{T(w')}))}$
	where $f(s,S) = 1$ if $s$ is maximal in $S$ and $0$ otherwise for any word $w'$.
	In contrast to $\softmax$, we have that the values of $f(\dotsb)$ are dependent of the overall context, namely the vector of all scorings
	$\score_{1,j}(\bs{x}, X^{0}_{T(w')})$. Compare $X^0_{T(u_1u_{j_1} \dotsb u_{j_{h_1}} u_{j_{h_2+1}} \dotsb u_2)}$ and $X^0_{T(u_1 \dotsb u_2)}$,
	both given by the additive-periodical embedding $\emb$. Via assumption,
	we have that each $u_{j_i}$ block also occurs in $u_1$ or $u_2$.
	In particular, this means every vector that occurs in $\emb(u_1 \dotsb u_2)$ does also occur in $\emb(u_1u_{j_1} \dotsb u_{j_{h_1}} u_{j_{h_2+1}} \dotsb u_2)$ and vice-versa.
	This implies that $f(\score_{1,j}(\bs{x}, \bs{x}'), \score_{1,j}(\bs{x}, X^0_{T(u_1u_{j_1} \dotsb u_{j_{h_1}} u_{j_{h_2+1}} \dotsb u_2)}))
	= f(\score_{1,j}(\bs{x}, \bs{x}'), \score_{1,j}(\bs{x}, X^0_{T(u_1 \dotsb u_2)}))$ for any scoring value $\score_{1,j}(\bs{x}, \bs{x}')$. 
	In combination with the assumption that $N_{u_1u_{j_1} \dotsb u_{j_{h_1}}}(\bs{x},\bs{x}',1,j) = N_{u_1u_{j_1} \dotsb u_{j_{h_2}}}(\bs{x},\bs{x}',1,j)$ and the observations above, we also get $N_{u_1u_{j_1} \dotsb u_{j_{h_1}} u_{j_{h_2+1}} \dotsb u_2}(\bs{x},\bs{x}',1,j) = N_{u_1 \dotsb u_2}(\bs{x},\bs{x}',1,j)$ in the $\hmax$ case.
	Next, consider the pooling functions. By definition, we have that $\pool_{1, j}(X^{0}_{T(w')}, \score_{1,j}(\bs{x}, X^{0}_{T(w')}))$ computes $\sum_{X^{0}_{T(w')}} \norm_{1,j}(\bs{x}, \bs{x}_{i'}, \score_{i,j}(\bs{x}, X^{0}_{T(w')})) (W\bs{x}_{i'})$ for any word $w'$. 
	Our previous arguments give that $N_{u_1u_{j_1} \dotsb u_{j_{h_1}} u_{j_{h_2+1}} \dotsb u_2}(\bs{x},\bs{x}',1,j) = N_{u_1 \dotsb u_2}(\bs{x},\bs{x}',1,j)$. In combination with 
	$P_{u_1u_{j_1} \dotsb u_{j_{h_1}}}(\bs{x},i,j) = P_{u_1u_{j_1} \dotsb u_{j_{h_2}}}(\bs{x},i,j)$ and (*), we immediately get that 
	$P_{u_1u_{j_1} \dotsb u_{j_{h_1}} u_{j_{h_2+1}} \dotsb u_2}(\bs{x},i,j) = P_{u_1 \dotsb u_2}(\bs{x},i,j)$ holds as well. 
	Next, consider layer $l_i$. The arguments are exactly the same as
	in the base case. However, we need to rely on the induction hypothesis. 
	Namely, we assume that all $\pool_{{i-1}, j}$ produce the same output in computation $T(u_1u_{j_1} \dotsb u_{j_{h_1}} u_{j_{h_2+1}} \dotsb u_2)$ and computation $T(u_1 \dotsb u_2)$. 
	This implies that all vectors present in $X^{i-1}_{T(u_1u_{j_1} \dotsb u_{j_{h_1}} u_{j_{h_2+1}} \dotsb u_2)}$ are also present in $X^{i-1}_{T(u_1 \dotsb u_2)}$ and vice-versa 
	and that the vectors corresponding to $u_{j_{h_2+1}} \dotsb u_2$ are equal in both computations.
\end{proof}

\begin{proof}[Proof of Lemma~\ref{sec:decidability;lem:shortword_softmax}]
	Let $T \in \periodTransCfix$ be an additive-periodical TE working over alphabet 
	$\Sigma$, having periodicity $p$, depth $L$, maximum width $H$, maximum dimensionality 
	$D$ and working over an FA $F$ using $b$ bits for binary encoding. 
	We use $V$ to denote the set of values representable in the fixed arithmetic that $T$ works over.
	Note that $|V| \leq 2^b$.
	Let $w \in \Sigma^+$ be a word such that $T(w) = 1$. We observe that there is $m \in \nats$ such
	that $w = u_1 \dotsb u_m u$ where $u_i \in \Sigma^p$ are blocks of symbols of length $p$ and $u \in \Sigma^{\leq p}$. 
	Our goal is to prove that a not necessarily connected subsequence of at most $2^{(|T|)^6}$ many 
	$p$-blocks $u_i$ from $u_1 \dotsb u_m$ is sufficient to ensure the same computation of $T$.
	In the case that $pm + p \leq 2^{(|T|)^6}$ we are done. Therefore, assume that $m > 2^{(|T|)^6}$.
	
	Let $U$ be the set of all unique $u_i$. We observe that $|U| \leq |\Sigma|^p$.
	Next, we fix some not necessarily connected but ordered subsequence $S = u_{j_0} u_{j_1} \dotsb u_{j_n} u_{j_{n+1}}$ 
	with $u_{j_0}=u_1 $, $j_i \in \{2, \dotsc, m\}$ and $u_{j_{n+1}} = u$ of $w$ such that each $u' \in U$ occurs exactly once.
	For the case that $u_1 = u$ we allow this specific block to occur twice in  $S$. 
	The assumption $m > 2^{\poly(|T|)}$ implies that $S \neq w$. This means that there are pairs 
	$(u_{j_h}, u_{j_{h+1}})$ in $S$ with some non-empty sequence of $p$-blocks $u_{j'_1} \dotsb u_{j'_l}$ in between.
	W.lo.g.\ assume $u_{j_0}$ and $u_{j_1}$ is such a pair. Our goal is to argue that there are at most $2^{(|T|)^5}$
	blocks from $u_{j'_1} \dotsb u_{j'_l}$ needed to ensure the same computation of $T$. Given that this argument works for all $|\Sigma|^p$ adjacent pairs in $S$, we are done. 
	
	Consider the computation $T(w)$. The additive-periodical embedding $\emb$ of $T$ implies that $\emb(w)$ includes
	at most $\Sigma p$ different vectors. Furthermore, from layer to layer equal vectors are mapped equally, which means
	that each $X_w^1, \dotsc , X_w^L$ contains at most $\Sigma p$ different vectors as well. This implies that the
	computation $T(w)$ induces at most $(L\Sigma p)^2 \times L \times H \leq (\Sigma p L^2 H)^2 \leq (\Sigma pLH)^4$ 
	different tuples $(\bs{x}, \bs{x}', i, j)$ where $\bs{x}, \bs{x}'$
	are vectors induced by $T(w)$ and $i \leq L, j \leq H$.
	Additionally, we have that for each value $N_{w}(\bs{x},\bs{x}',i,j)$ and $P_{w}(\bs{x},i,j)$, as defined in the beginning of this section, there are at most 
	$|V^D| \leq 2^{bD}$ possibilities. Simple combinatorics, namely the pigeon hole principle, states that in the 
	increasing sequence $u_{j'_1}, u_{j'_2}, \dots$ there must be points $h_1$ and $h_2$ with 
	$h_1 \leq 2^{bD(\Sigma pLH)^4} \leq 2^{(|T|)^5}$ such that for all tuples
	$(\bs{x}, \bs{x}', i, j)$ induced by $T(w)$ we have that $N_{u_{j_0}u_{j'_1} \dotsb u_{j'_{h_1}}}(\bs{x},\bs{x}',i,j) = N_{u_{j_0}u_{j'_1} \dotsb u_{j'_{h_2}}}(\bs{x},\bs{x}',i,j)$ and $P_{u_{j_0}u_{j'_1} \dotsb u_{j'_{h_1}}}(\bs{x},i,j) = P_{u_{j_0}u_{j'_1} \dotsb u_{j'_{h_2}}}(\bs{x}, i,j)$.
	Now, Lemma~\ref{app:dec;lem:cutout} states that this implies 
	$N_{u_{j_0}u_{j'_1} \dotsb u_{j'_{h_1}}u_{j'_{h_2+1}} \dotsb u_{j_1} \dotsb u}(\bs{x},\bs{x}',i,j) = N_{w}(\bs{x},\bs{x}',i,j)$ 
	and $P_{u_{j_0}u_{j'_1} \dotsb u_{j'_{h_1}}u_{j'_{h_2+1}} \dotsb u_{j_1} \dotsb u}(\bs{x},i,j) = P_{w}(\bs{x}, i,j)$.
	However, this implies that the subsequence 
	$u_{j'_{h_1+1}} \dotsb u_{j'_{h_2}}$ has no influence in the computation of $T$ on $w$ and, thus, can be left out. 
	As we can argue this for every such cycle occurring in
	$u_{j'_1} \dotsb u_{j'_l}$, we get the desired bound of 
	$2^{(|T|)^5}$.
\end{proof}

\begin{proof}[Proof of Theorem~\ref{sec:decidability;thm:nexptimehard}]
	First, we argue the decidability of $\empt[\transCfix]$. Assume that $T \in \transCfix$ 
	with an arbitrary embedding $\emb$ is given that operates in a fixed-width arithmetic using $b$ bits 
	for representing numbers and wrap-around to handle overflow. 
    Then, $\emb$ is periodic with periodicity $p \leq 2^b$, simply due to the fact that positions
    $i$ in some word $w$ can only be exactly represented up to magnitude $2^b$. Therefore, the same
    arguments as used in Theorem~\ref{sec:decidability;thm:nexptime} apply here. Note that this
	does not imply \NEXPTIME-membership of $\empt[\transCfix]$, due to the fact that the period is
    exponential in $b$. Analogously, in a saturating scenario, we have that $\emb$ has a finite
    prefix of length at most $2^b$ and is periodic with periodicity $1$ afterwards. Here, the small-word
	property used in Theorem~\ref{sec:decidability;thm:nexptime} follows the same line of reasoning, 
	with the difference that either the finite prefix is sufficient as a witness, or the finite prefix
    followed by an exponentially bounded suffix, whose existence follows from the same arguments as in
	Lemma~\ref{sec:decidability;lem:shortword_softmax} with periodicity $p=1$.

	Second, we argue the \NEXPTIME-hardness. We prove the statement via reduction from $\bOTWP_{\mathsf{bin}}$.  
	Let $\mathcal{S} = (S,H,V,t_I,t_F)$ and $n \geq 1$ be an instance of $\bOTWP_{\mathsf{bin}}$.
	We construct an TE $T_{\mathcal{S},n} \in \transCfix$ working over some FA $F$ with 
	$T_{\mathcal{S},n}(w) = 1$ if and only if $w \in S^+$ witnesses the validity of the $\bOTWP_{\mathsf{bin}}$ 
	instance $(\mathcal{S}, n)$.  
	
	Next, let $T_{\mathcal{S},n}$ be built exactly like $T_{\mathcal{S}}$ in the proof of 
	Theorem~\ref{sec:undec}, but with the following structural adjustments.
	In layer $l_3$ we adjust $\comb_3$ to be
	$\comb_3 = N_3 \para N_e \para N_f$ where $N_3$ is specified
	as in the proof of Theorem~\ref{sec:undec},  $N_e = N_\rightarrow \circ (N^{x_{1,2}, x_{3,2}}_= \para N^{x_{1,3}}_{=n})$ and $N_f = N^{x_{1,2}}_{\neq \frac{(n+1)((n+1)+1)}{2}+1}$ where
	$N_{\neq t}$ is analogous to the construction of 
	$N_{= t}$ given in 
	Lemma~\ref{app:undec_proofs;lem:tilegadgets}. 
	Furthermore, we adjust $\comb_4$ in layer $l_4$ to be represented by the FNN $\relu(x_3 + \dotsb + x_8 + x_9)$. 
	We refer to the gadgets described in Lemma~\ref{app:undec_proofs;lem:gadgets} and Lemma~\ref{app:undec_proofs;lem:tilegadgets} as well as the proof of Theorem~\ref{sec:undec;thm:undec}
	for further details. 
	
	Consider the adjustment in $l_3$. FNN $N_e$ in $\comb_3$ ensures that 
	$T_{\mathcal{S},n}(w)=1$ only if the row index corresponding to 
	the  last symbol is equal to $n$. Note that $N_3$ checks whether 
	row and column index corresponding to the last symbol are equal. 
	Additionally, $N_f$ checks if there is no id
	equal to $\frac{(n+1)((n+1)+1)}{2}+1$. This corresponds
	to the position id of the successor of the vector representing tile $(n,n)$.
	Furthermore, the adjustment of $\comb_4$ considers the output of $N_e$ and $N_f$ in addition to the outputs of 
	$N_3$. In summary, we have that $T_{\mathcal{S},n}$ only outputs $1$ given $w$ if the word length is such that 
	the row index corresponding to the position of the last symbol of $w$ in a respective octant tiling is equal to $n$ 
	(ensured by $N_e$), that $w$ is at most of length $\frac{(n+1)((n+1)+1)}{2}$ (ensured by $N_f$) and if $w$ represents a valid encoded tiling (the remaining parts of $T_{\mathcal{S},n})$.
	
	Additionally, we need to argue that $T_{\mathcal{S},n}$ works as intended,
	despite the fact that it is limited by some FA $F$ using a representation size that is at most logarithmic in $n$. 
	These arguments follow the exact same line as in the proof of Theorem~\ref{sec:undec;thm:logundec}, but using FA $F$ that uses 
	$m = \rdown{6\log(\max(|S|, n))}+2$ bits and handles overflow using saturation. The reason for the larger representation size is that words $w$ 
	representing a valid encoded tiling ending at position $(n,n)$ are of length $|w| = \frac{(n+1)((n+1)+1)}{2} \leq n^2$. 
	Thus, we use $\rdown{4\log(n)}+1$ integer bits to be able to represent a sum $\sum_{j=0}^{i} j = \frac{i(i+1)}{2} \leq i^2$ for all 
	$i \leq n^2$ and $\rdown{2\log(n)}+1$ fractional bits to uniquely represent fraction $\frac{1}{l}$ for $l \leq n$. For detail see the proof
	of Theorem~\ref{sec:undec;thm:logundec}. Furthermore, the fact that we use $\rdown{4\log(n)}+1$ bits to encode integers and 
	that $F$ handles overflow using saturation ensures that $N_f$ works as intended: we have that $\frac{(n+1)((n+1)+1)}{2} + 1 < n^4$
	and, thus, we have that the id $\frac{(n+1)((n+1)+1)}{2}+1$ occurs at most once, independent of the length of $w$ as it is not
	the point where $F$ enforces saturation on the positional embedding. 
	Thus, $\att_{\text{self}}$ works for this position as intended and then $N_f$ checks the property described above correctly.
	
	The argument that $T_{\mathcal{S},n}$ can be built in polynomial time is a straightforward implication from the arguments for Theorem~\ref{sec:dec;thm:boundedword} and the fact that $N_e$ and $N_f$ are a small gadgets with maximum parameter quadratic in $n$, which can be represented using a logarithmic amount of bits.
\end{proof}

\end{document}